\documentclass[11pt,a4paper]{article}
\usepackage[T1]{fontenc}
\usepackage{jheppub}
\usepackage{lmodern,microtype}
\usepackage{amsmath,amssymb,amsthm,bm}
\usepackage{booktabs,array,graphicx}
\usepackage{xcolor}
\newif\ifshowchanges\showchangestrue

\numberwithin{equation}{section}
\newtheorem{theorem}{Theorem}[section]
\newtheorem{proposition}[theorem]{Proposition}
\newtheorem{lemma}[theorem]{Lemma}
\theoremstyle{remark}
\newtheorem{remark}[theorem]{Remark}
\newcommand{\dd}{\mathrm d}
\newcommand{\I}{\mathcal I_3}
\newcommand{\Rcal}{\mathcal R}
\newcommand{\Lcal}{\mathcal L} 
\newcommand{\Bcal}{\mathcal B}
\newcommand{\eps}{\varepsilon}
\newcommand{\bq}{\bm q}
\newcommand{\ee}{\bm e}

\newcommand{\TT}{\mathcal T}

\newcommand{\ovl}{\overline}
\allowdisplaybreaks[1]

\title{Exact Two-Loop Effective Potentials in de Sitter Space via Finite-Difference Methods}
\author[1,2]{Bruce A. Bassett}
\author[3,4]{Sergio L. Cacciatori}
\author[3,4]{Ugo Moschella}
\affiliation[1]{Wits MIND Institute and School of Computer Science and Applied Mathematics, University of the Witwatersrand, Johannesburg, South Africa}
\affiliation[2]{School for Data Science and Computational Thinking, University of Stellenbosch, South Africa}
\affiliation[3]{DiSAT, Universit\`a dell'Insubria, Via Valleggio 11, Como}
\affiliation[4]{INFN, Sez. di Milano, Via Celoria 16, 20133 Milano}
\emailAdd{bruce.bassett@uct.ac.za}
\emailAdd{sergio.cacciatori@uninsubria.it}
\emailAdd{ugo.moschella@uninsubria.it}

\abstract{We compute the renormalised two-loop effective potential of the
$O(N)$-symmetric scalar theory on four-dimensional de Sitter space. The main
technical ingredient is the two-loop watermelon vacuum integral. Although a
closed-form representation of this integral was previously available,
extracting its finite part near four dimensions remained exceedingly difficult. We use
exact finite-difference equations in the propagator spectral labels to obtain
a representation in which the ultraviolet poles are isolated in Gamma-function
factors and the remaining series are regular near four dimensions. We prove
the boundary condition selecting the physical solution, extract the finite
part, and combine the result with the tadpole reconstruction of the potential.
The calculation includes the subtraction of subdivergences and the derivation
of the matter-sector renormalisation-group functions. Holding the
field-dependent masses and scalar curvature fixed makes the spectral labels
regulator-dependent. We track this dependence explicitly and show that its
derivative terms cancel after the forest subtraction. The resulting potential
is exact in the mass-to-curvature ratios at two-loop order, away from
propagator poles and subject to the usual limitations of perturbation theory. This is the first exact two-loop determination of the effective
potential in four-dimensional de Sitter space.}

\begin{document}
\maketitle

\section{Introduction}

The effective potential is the standard tool for studying the vacuum
structure of a quantum field theory. Since the work of Coleman and Weinberg and
of Jackiw \cite{ColemanWeinberg1973,Jackiw1974}, it has been used to analyse
radiative symmetry breaking, vacuum stability and phase transitions, and its
radiative corrections often decide the qualitative picture: whether a symmetry
is broken, whether a vacuum is stable, and at which scale.

In de Sitter space these questions are of direct physical interest. De
Sitter space is the geometry of inflation and, asymptotically, of the late-time
accelerating Universe, and radiative corrections in this background govern
vacuum stability, curvature-induced symmetry breaking and the dynamics of light
scalar fields during cosmological evolution
\cite{Shore,Allen1983,Starobinsky1986,StarobinskyYokoyama1994,MarkkanenRajantieStopyra2018}.
Curvature enters on the same footing as the masses, so a reliable potential
should be exact in both. Perturbation theory in de Sitter space is, however,
considerably harder than in flat spacetime, and beyond one loop most results rely
on approximations, asymptotic expansions, numerical methods or partial
resummations \cite{Serreau2011,LopezNacirMazzitelliTrombetta2016}. Exact
effective potentials are known at one loop, and at two loops in three dimensions
\cite{ACM}.

In flat space the two-loop effective potential has long been known: Ford,
Jack and Jones obtained it for the Standard Model \cite{FordJackJones1992}. At
two loops the new ingredient is the sunset, or watermelon, vacuum diagram.
Rather than integrating it directly, they used integration by parts to derive
differential equations for it in the masses, and solved them, obtaining its
ultraviolet poles and its finite part in closed form. Integration-by-parts
identities \cite{Chetyrkin:1981qh,Laporta:2000dsw}, which follow from the
translation invariance of momentum space, have since become the standard route
to such computations.

None of this carries over to de Sitter space. There are no global
momenta, and integration-by-parts identities in their flat-space form are not
available. The watermelon itself can be evaluated: it was obtained in closed form
in Ref.~\cite{CEM2024} through the K\"all\'en--Lehmann representation of the
product of two propagators \cite{bros1,bros2}. In four dimensions, however,
extracting its finite part from that representation proved exceedingly
difficult, and the two-loop effective potential in four-dimensional de Sitter
space remained out of reach. Curved-space structures related to
integration-by-parts identities have recently been explored from several
directions \cite{Arkani,Chen,ChenFengQinTao,DiakonovMorozov,CostaHansen2018,Alaverdian,Carmi2020,Carmi2021,Carmi2024},
but none had produced the finite parts this computation requires.

The obstruction was removed in the companion Letter \cite{BCM}, which
introduced the curved-space counterpart of integration-by-parts identities:
exact finite-difference equations in the spectral labels of the lines, generated
by two local, coordinate-free identities, propagator contiguity and Green's
identity, shared by de Sitter and anti-de Sitter space. Where the flat-space
computation relies on differential equations in the masses, here the masses
enter through the spectral labels of the propagators, which can only be shifted
by integers, and the equations are finite-difference equations in those labels.
For the watermelon they give an exact mass-shift recurrence. Its physical
solution, selected by the appropriate boundary conditions, is a new
representation in which all ultraviolet poles sit in Gamma-function factors
while the remaining hypergeometric series are regular near four dimensions, so
that both the poles and the finite part are explicit. This makes possible, for
the first time, a computation in de Sitter space parallel to the flat-space one;
the same equations hold in anti-de Sitter space, which we do not treat here.

In this paper we carry out that computation for the $O(N)$-symmetric scalar
theory on four-dimensional de Sitter space. We give the complete derivation of
the watermelon on the sphere, including a proof that the boundary conditions
select the physical amplitude, and its expansion near four dimensions. The
calculation is performed in dimensional regularisation and includes the
complete subtraction of subdivergences. A feature without flat-space analogue
is that, when physical masses and curvature are held fixed, the spectral labels
of the propagators depend on the regulator; we show that this dependence
modifies the finite part of the potential while leaving its ultraviolet
structure unchanged. Our main results are:
\begin{itemize}
\item the finite-difference equations satisfied by the watermelon on the
sphere, with complete proofs;
\item the physical solution of the recurrence, with a proof of the boundary
condition;
\item the dimensional expansion of the watermelon near four dimensions,
including its finite part;
\item the fully renormalised two-loop effective potential of the $O(N)$
model on de Sitter space, exact in masses and curvature;
\item the corresponding renormalisation-group functions.
\end{itemize}To our knowledge, this is the first exact two-loop effective potential in
four-dimensional de Sitter space.

The paper is organised as follows. In Section~\ref{sec:2} we introduce the
$O(N)$ model on the Euclidean de Sitter sphere and fix our notation.
Section~\ref{sec:1loop} reviews the one-loop effective potential. In
Section~\ref{sec:watermelon} we derive the finite-difference equations for the
watermelon and their exact solution. Section~\ref{sec:5} develops the
four-dimensional expansion and renormalisation, and presents the two-loop
effective potential. In Section~\ref{sec:6} we analyse the ultraviolet structure
and derive the beta functions. Technical details on the finite-difference
solution are collected in Appendix~\ref{sec:abel}.

\section{The \texorpdfstring{$O(N)$}{O(N)} model on the Euclidean de Sitter sphere}\label{sec:2}
Let $M_{d+1}$ be real $(d+1)$-dimensional Minkowski spacetime and
$M_{d+1}^{(c)}$ its complexification. In a chosen Lorentz frame,
\begin{equation}
 Z_1\cdot Z_2=Z_1^0Z_2^0-Z_1^1Z_2^1-\cdots-Z_1^dZ_2^d.
\end{equation}
The complex de Sitter manifold is the hyperboloid
\begin{equation}
 dS_d^{(c)}=\{Z\in M_{d+1}^{(c)}:Z\cdot Z=-R^2\}.
\end{equation}
We initially set $R=1$; physical units are restored in
Section~\ref{sec:5}. The invariant complex variable is
\begin{equation}z=Z_1\cdot Z_2.\end{equation}
The future and past tuboids are
\begin{equation}\label{s.2.1}
 \TT_\pm=\{Z={x+iy}\in dS_d^{(c)}:{y}\in\pm V_+\},
\end{equation}
where $V_+$ is the ambient future timelike cone. The de Sitter spacetime
has no global timelike Killing vector; the spectral condition cannot be formulated in the usual way but may be replaced 
by the  normal analyticity property~\cite{bgm,bm}: the two-point distributions
are boundary values of functions holomorphic in $\TT_-\times\TT_+$.
Together with de Sitter invariance and the canonical commutation relations,
this fixes the massive scalar two-point function~\cite{bgm,bm,CEM2024}, with $h=\frac{d-1}{2}$:
\begin{align}
 W^d_\nu(Z_1,Z_2)
 &=\frac{\Gamma({h}+i\nu)\Gamma({h}-i\nu)}
 {2(2\pi)^{d/2}}(z^2-1)^{-\frac{d-2}{4}}
 P_{-\frac12+i\nu}^{-\frac{d-2}{2}}(z)\label{leg2p}\\
 &=\frac{\Gamma({h}+i\nu)\Gamma({h}-i\nu)}
 {(4\pi)^{d/2}\Gamma(d/2)}
 {}_2F_1\!\left({h}+i\nu,{h}-i\nu;
 \frac d2;\frac{1-z}{2}\right).\label{wig}
\end{align}
The spectral parameter is related to the squared mass by
\begin{equation}\label{cmass}{M^2=h^2+\nu^2}.\end{equation}
For real positive squared mass, $\nu$ is either real (the principal series)
or purely imaginary with $|\nu|<{h}$ (the complementary series).
A discrete series of acceptable quantum field theories with negative
squared mass also exists~\cite{Tachyons1,Tachyons2}.

The Schwinger function is the restriction to the Euclidean sphere
\begin{equation}\label{s.2}
 S^d=\{Z\in dS_d^{(c)}:Z^0=iX^0,\ Z^i=X^i,\ X^\mu\in\mathbb R\}.
\end{equation}
On this section, $X\in S^d\subset\mathbb R^{d+1}$ and 
\begin{equation}Z_1\cdot Z_2=-X_1\cdot_E X_2=-\cos s,\end{equation}
where $s$ is the geodesic separation. In the following we use the notations
\begin{equation}
 {q=i\nu,\qquad u=X_1\cdot_E X_2=\cos s.}
\end{equation}
With all the above understood the propagator is
\begin{equation}\label{leg6i}
 G^{(d)}_{{q}}({u}) = \frac{\Gamma({h+q})\Gamma({h-q})}{2 (2\pi)^{d/2}}\,
 {(1-u^2)^{-\frac{d-2}{4}}}\,{\bf P}_{{-\frac12+q}}^{-\frac{d-2}2}({-u}),
\end{equation}
where {${\bf P}^\alpha_\beta(x)$} is the so called ``Legendre function on the cut''~\cite{bateman} or Ferrers function of the first kind.\footnote{It is important to keep in mind that Ferrers functions {${\bf P}^\alpha_\beta(x)$} and Legendre functions {$P^\alpha_\beta(x)$} are holomorphic in
different cut-planes; as regards Ferrers function this is
\begin{equation}
\Delta_2 = {\bf C}\setminus {\{(-\infty,-1]\cup [1,\infty)\}}.
\label{d.20}\end{equation}}
The Euclidean propagator is equivalently written as follows:
\begin{equation}\label{eq:Euclidean-propagator}
 G_q^{(d)}(u)=
 \frac{\Gamma(h+q)\Gamma(h-q)}{(4\pi)^{d/2}\Gamma(d/2)}
 {}_2F_1\!\left(h+q,h-q;\frac d2;\frac{1+u}{2}\right),
 \qquad M^2=h^2-q^2.
\end{equation}
When the dimension is understood, we abbreviate $G_q^{(d)}$ to $G_q$.

Consider a scalar field $\phi:S^d\to\mathbb R^N$ with action
\begin{equation}
 S[\phi]=\int_{S^d}\left[
 \frac12\langle\partial_\mu\phi\mid\partial^\mu\phi\rangle
 +\frac{m_0^2+\xi\Rcal}{2}\langle\phi\mid\phi\rangle
 +\frac{c_0}{4}\langle\phi\mid\phi\rangle^2+\Lcal_g
 \right]\,{\dd X},
\end{equation}
where $m_0,c_0$ are bare parameters, {$\dd X$ is the invariant measure on $S^d$,} and
\begin{equation}
 \Lcal_g=\rho+\kappa\Rcal+
 \left(b_R+b_E\frac{(d-2)(d-3)}{d(d-1)}\right)\Rcal^2.
\end{equation}
We expand about $\ovl\phi=\varphi{\mathbf u_0}$, where $\varphi$ is constant and
{$\mathbf u_0$} is a unit vector. In an orthonormal basis {$\{\mathbf u_j\}_{j=0}^{N-1}$ of $\mathbb R^N$ whose first element is $\mathbf u_0$},
\begin{equation}
 \phi=(\varphi+\psi_0){\mathbf u_0}+\sum_{j=1}^{N-1}\psi_j{\mathbf u_j}.
\end{equation}
The action becomes~\cite{CEM2024}
\begin{equation}
 S[\phi]=\int_{S^d}\left[\frac{m_0^2+\xi\Rcal}{2}\varphi^2
 +\frac{c_0}{4}\varphi^4+\Lcal_g\right]\,{\dd X}
 +\sum_{a=1}^4S_a[\psi;\varphi],
\end{equation}
with
\begin{align}
 S_1&=\int_{S^d}\big[(m_0^2+\xi\Rcal)\varphi\psi_0
 +c_0\varphi^3\psi_0\big]\,{\dd X},\\
 S_2&=\frac12\int_{S^d}\left[
 \sum_{j=0}^{N-1}\partial_\mu\psi_j\partial^\mu\psi_j
 +M_0^2\psi_0^2+M_1^2\sum_{j=1}^{N-1}\psi_j^2
 \right]\,{\dd X},\\
 S_3&=\int_{S^d}c_0\varphi\psi_0\sum_{j=0}^{N-1}\psi_j^2\,{\dd X},\\
 S_4&=\frac{c_0}{4}\int_{S^d}\left(\sum_{j=0}^{N-1}\psi_j^2\right)^2
 \,{\dd X},
\end{align}
where
\begin{equation}
 M_0^2=m_0^2+\xi\Rcal+3c_0\varphi^2,\qquad
 M_1^2=m_0^2+\xi\Rcal+c_0\varphi^2.
\end{equation}
These are the radial and transverse Hessian masses. At a general background,
the transverse modes need not be massless Goldstone modes.

For a constant configuration, the Legendre transform gives the effective
potential in the form
\begin{align}\label{313}
 V(\varphi)={}&\Lcal_g+\frac{m_0^2+\xi\Rcal}{2}\varphi^2
 +\frac{c_0}{4}\varphi^4
 -\frac1{V_d}\log\int\!\left[\prod_jD\psi_j\right]e^{-S_2}
 \nonumber\\
 &-\frac1{V_d}\left.\log
 \frac{\int[\prod_jD\psi_j]e^{-S_2-S_3-S_4}}
 {\int[\prod_jD\psi_j]e^{-S_2}}\right|_{\mathrm{1PI}}.
\end{align}
Here $V_d$ is the volume of $S^d$. The effective action is proportional to
$V_d$, and $V$ is the action per unit volume. The formal measure is defined
through dimensional regularisation. Its normalisation cancels in the ratio;
any remaining field-independent constant is included in the vacuum-energy
normalisation. The Gaussian integral gives the one-loop determinants. The
logarithm of the ratio generates connected vacuum diagrams, while the
Legendre transform removes $S_1$ and the one-particle-reducible contributions.
At two loops the retained topologies are the figure-eight generated by $S_4$
and the watermelon generated by two cubic vertices.

\section{The one-loop bare potential}\label{sec:1loop}
The tree-level plus one-loop potential can be written as
\begin{equation}
 V_{1\text{-loop}}(\varphi)=\Lcal_g+\frac{m_0^2+\xi\Rcal}{2}\varphi^2
 +\frac{c_0}{4}\varphi^4+V_0(\varphi)+(N-1)V_1(\varphi).
\end{equation}
For a mode of mass $M$,
\begin{equation}
 V(M)=\frac1{2V_d}\log\det(-\Delta+M^2),\qquad
 \frac{\partial V}{\partial M^2}
 =\frac1{2V_d}\operatorname{Tr}(-\Delta+M^2)^{-1}=\frac12T_d(q).
\end{equation}
The unregularised determinant is divergent. Its mass derivative, however,
converges for ${\rm Re} (d)<2$, where the propagator is finite at coincidence.
In that domain the tadpole equation~\cite{ACM} determines the mass dependence
of $V_j$. Dimensional regularisation continues the result to $d$ near four.
A mass-independent curvature normalisation remains to be specified.
At coincidence $u=1$, Gauss's summation theorem applied to
Eq.~\eqref{eq:Euclidean-propagator} gives
\begin{equation}\label{eq:tadpole}
 T_d(q)=\frac{\Gamma(1-d/2)}{(4\pi)^{d/2}}
 \frac{\Gamma(h+q)\Gamma(h-q)}
 {\Gamma(1/2+q)\Gamma(1/2-q)}.
\end{equation}
We write $T_q=T_d(q)$ when the dimension is understood. In the
$\nu$-parameterisation of the one-loop integration below, the same
quantity is $T_d(i\nu)$.
Its meromorphic continuation defines the dimensionally regularised tadpole.
For generic non-even $d$, the short-distance propagator splits into a
singular Frobenius branch proportional to $r^{2-d}$ and a branch analytic
at $r=0$; $T_d$ is the constant coefficient of the latter. Their resonance
at even dimensions is reflected in the poles of $\Gamma(1-d/2)$.
Near four dimensions the pole is proportional to
$\nu^2+1/4=M^2-\Rcal/6$ at $d=4$ in unit-radius variables.
Since $\dd M^2=2\nu\,\dd\nu$ at fixed $d$, integration gives
\begin{equation}\label{325}
 V_j(\varphi)={\frac{\Gamma(\tfrac 12-h)}{(4\pi)^{d/2}}}
 \int^{\nu_j(\varphi)}
 \frac{\Gamma({h}+ix)\Gamma({h}-ix)}
 {\Gamma(\frac12+ix)\Gamma(\frac12-ix)}x\,\dd x,
\end{equation}
where $\nu_j(\varphi)=\sqrt{M_j^2(\varphi)-{h^2}}$.
For $d=4-2\eps$, at fixed $\nu$,
\begin{align}\label{T4}
 T_{4-2\eps}(i\nu)={}&-
 \frac{(1+4\nu^2)(\eps^{-1}+1-\gamma_E+\log4\pi)}{64\pi^2}
 \nonumber\\
 &+\frac{(\nu^2+1/4)
 [\psi(-\frac12-i\nu)+\psi(-\frac12+i\nu)]}{16\pi^2}+O(\eps).
\end{align}
Integrating yields
\begin{align}\label{3.30}
 V_j^{(4)}(\varphi)={}&
 \left[-\frac{(x^4+x^2/2)(\eps^{-1}-\gamma_E+1+\log4\pi)
 +(x^2+1/4)^2}{64\pi^2}\right.\nonumber\\
 &\left.\quad+\frac{(x^2+1/4)^2}{64\pi^2}
 [\psi(-\tfrac12+ix)+\psi(-\tfrac12-ix)]
 +\frac{B(x)}{16\pi^2}\right]_{x=\sqrt{M_j^2(\varphi)-9/4}}+O(\eps),
\end{align}
where
\begin{align}\label{Bfunction}
 B(x)&=\frac{(x^2+1/4)^2}{4}
 -\int^x\frac i4(y^2+1/4)^2
 [\psi'(-\tfrac12+iy)-\psi'(-\tfrac12-iy)]\,\dd y\nonumber\\
 &=-\frac12\int^{x^2+1/4}\!\dd y\,y^2
 \left[\sum_{n=0}^\infty\frac{2n-1}{2(n(n-1)+y)^2}-\frac1y\right].
\end{align}
An alternative expression in terms of antiderivatives of $\psi$ is
\cite{VicenteGarcia-Consuegra:2025two}
\begin{align}
 B(x)={}&-\frac{ix(4x^2+1)}2
 [\log\Gamma(-\tfrac12-ix)-\log\Gamma(-\tfrac12+ix)]\nonumber\\
 &+12[\psi^{(-4)}(-\tfrac12-ix)+\psi^{(-4)}(-\tfrac12+ix)]\nonumber\\
 &+12ix[\psi^{(-3)}(-\tfrac12-ix)-\psi^{(-3)}(-\tfrac12+ix)]\nonumber\\
 &-\frac{12x^2+1}{2}
 [\psi^{(-2)}(-\tfrac12-ix)+\psi^{(-2)}(-\tfrac12+ix)]\nonumber\\
 &+\frac{(4x^2+1)^2}{32}
 [\psi(-\tfrac12-ix)+\psi(-\tfrac12+ix)].
\end{align}
The antiderivative constants are fixed consistently with Eq.~\eqref{Bfunction}.
\begin{remark}
The one-loop calculation was considered in Ref.~\cite{ACM} at fixed $H$.
Here we hold the field-dependent masses and $\Rcal$ fixed during dimensional
continuation. Then $H^2=\Rcal/[d(d-1)]$ and the spectral labels depend on
$\eps$. We track this dependence explicitly below. At one loop it changes
only the polynomial part of the potential in the prescription used here.
\end{remark}
For the two-loop potential we also require the watermelon integral.
Although a closed-form representation was obtained in Ref.~\cite{CEM2024},
extracting its four-dimensional finite part was difficult. The
finite-difference method of Ref.~\cite{BCM} provides the representation
used in the next section.

\section{The watermelon from finite-difference equations}\label{sec:watermelon}
On the unit Euclidean sphere, define
\begin{equation}
 \I(\bq)=\int_{S^d}G_{q_1}(X,Y)G_{q_2}(X,Y)G_{q_3}(X,Y)\,\dd Y,
 \qquad \bq=(q_1,q_2,q_3),
\end{equation}
where $G_q$ is the propagator defined in
Eq.~\eqref{eq:Euclidean-propagator}, with $u=X\cdot_E Y$.
We retain $h=(d-1)/2$ and introduce
\begin{equation}\label{eq:conv}
 a=h-1,\qquad \mu_q=h^2-q^2,\qquad \kappa_q=a^2-q^2.
\end{equation}
Our first aim is the fixed-label Laurent expansion
\begin{equation}\label{eq:laurent}
 (16\pi^2)^2\I(\bq)=\frac{C_{-2}(\bq)}{\eps^2}
 +\frac{C_{-1}(\bq)}{\eps}+C_0(\bq)+O(\eps),\qquad d=4-2\eps.
\end{equation}
Fixed labels do not imply fixed masses, since $\mu_q=h^2-q^2$ depends on $d$.
The tadpole $T_q=T_d(q)$, defined in Eq.~\eqref{eq:tadpole}, obeys
\begin{equation}\label{eq:Tshift}
 T_{q+1}=\frac{h+q}{q-a}\,T_q.
\end{equation}
We denote the canonical basis of $\mathbb R^3$ by $\{\ee_1,\ee_2,\ee_3\}$.
All identities are initially understood away from propagator poles and
poles of the separated representation. Removable singularities are taken
in the combined expression.

\subsection{The recurrence and its solution}\label{sec:rec}
\begin{proposition}[Mass-shift recurrence]\label{prop:rec}
For $\{i,j,k\}=\{1,2,3\}$,
\begin{equation}\label{eq:rec}
 P_k^-\I(\bq+\ee_k)-P_k^+\I(\bq-\ee_k)=\mathcal S_k,
\end{equation}
where
\begin{equation}\label{eq:PR}
 P_k^\pm=\prod_{\epsilon_1,\epsilon_2=\pm1}
 (a\pm q_k+\epsilon_1q_i+\epsilon_2q_j)
 {{}=\lambda-2a(a\pm q_k)\,\frac{\partial\lambda}{\partial\kappa_k},}
\end{equation}
{with
$\lambda=\kappa_1^2+\kappa_2^2+\kappa_3^2-2(\kappa_1\kappa_2+\kappa_2\kappa_3+\kappa_3\kappa_1)$
the K\"all\'en function of $\kappa_i=\kappa_{q_i}$,}
and
\begin{equation}\label{eq:RR}
 \mathcal S_k=2q_k(2h-1)\left[
 2T_iT_j-(T_i+T_j)T_k
 +\frac{\kappa_i-\kappa_j}{\kappa_k}(T_i-T_j)T_k\right].
\end{equation}
Here $T_i=T_{q_i}$.
\end{proposition}
\begin{proof}
It suffices to take $(i,j,k)=(1,2,3)$; the other channels follow by
permutation. Primes denote derivatives with respect to $u$. The proof has
three steps: the shifted integrals are reduced to a single radial integral,
its integrand is shown to be a total derivative, and the boundary values of
the corresponding current are evaluated.

\emph{Step 1: reduction to one radial integral.} The integrand of $\I$
depends on $Y$ only through $u=X\cdot_EY$, so that
\begin{equation}\label{eq:weights}
 \I(\bq)=\omega_d\int_{-1}^1w_d(u)\,G_{q_1}G_{q_2}G_{q_3}\,\dd u,\qquad
 \omega_d=\frac{2\pi^{d/2}}{\Gamma(d/2)},\qquad
 w_d(u)=(1-u^2)^{d/2-1},
\end{equation}
with $\omega_d$ the volume of the unit $S^{d-1}$. The contiguity relations
\begin{align}\label{eq:ladders}
 (q-a)G_{q+1}&=(q+h)uG_q-(1-u^2)G_q',\nonumber\\
 (q+a)G_{q-1}&=(q-h)uG_q+(1-u^2)G_q',
\end{align}
which follow from the hypergeometric contiguous relations and the
normalisation of the coincidence behaviour, express $G_{q_3\pm1}$ through
$G_{q_3}$ and $G_{q_3}'$. Inserting them in Eq.~\eqref{eq:weights}, gives
\begin{equation}\label{eq:AB-step}
 P_3^-\I(\bq+\ee_3)-P_3^+\I(\bq-\ee_3)
 =\omega_d\int_{-1}^1w_d(u)\big[
 A\,uG_{q_1}G_{q_2}G_{q_3}
 +B\,(1-u^2)G_{q_1}G_{q_2}G_{q_3}'\big]\,\dd u,
\end{equation}
where
\begin{equation}
 A=\frac{(q_3+h)P_3^-}{q_3-a}-\frac{(q_3-h)P_3^+}{q_3+a},\qquad
 B=-\frac{P_3^-}{q_3-a}-\frac{P_3^+}{q_3+a}.
\end{equation}

\emph{Step 2: the integrand is a total derivative.} Consider the current
\begin{align}\label{eq:certJ}
 J(u)=w_d(u)(1-u^2)\Bigg[&b_0G_{q_1}G_{q_2}G_{q_3}
 +u\sum_i b_iG_{q_i}'\prod_{j\ne i}G_{q_j}\nonumber\\
 &+(1-u^2)\sum_{i<j}c_{ij}G_{q_i}'G_{q_j}'G_{q_k}\Bigg],
\end{align}
with constant coefficients $b_0$, $b_i$, $c_{ij}=c_{ji}$, and $k$ the third
index. We differentiate $J$ and eliminate the second derivatives with the
propagator equation, valid away from coincidence,
\begin{equation}
 (1-u^2)G_q''-d\,uG_q'-\mu_qG_q=0 .
\end{equation}
Then $J'$ is $w_d$ times a combination of $uG_{q_1}G_{q_2}G_{q_3}$, of the
products with one, two and three derivatives, and it equals the integrand of
Eq.~\eqref{eq:AB-step} if and only if
\begin{gather}
 -d\,b_0+\sum_i\mu_{q_i}b_i=A,\qquad
 b_0+b_i+\sum_{j\ne i}\mu_{q_j}c_{ij}=B\,\delta_{i3},\nonumber\\
 (d-2)c_{ij}+b_i+b_j=0,\qquad \sum_{i<j}c_{ij}=0 .
\end{gather}
The first two conditions match the terms without derivatives and with one
derivative; the last two cancel the terms with two and three derivatives.
The system is overdetermined. With $P_3^\pm$ given by Eq.~\eqref{eq:PR}, one
verifies by direct substitution that it is solved by
\begin{align}\label{eq:bs}
 b_0&=2q_3(\kappa_3-\kappa_1-\kappa_2),\nonumber\\
 b_1&=2q_3(d-2)\left(-1-\frac{\kappa_1-\kappa_2}{\kappa_3}\right),&
 b_2&=2q_3(d-2)\left(-1+\frac{\kappa_1-\kappa_2}{\kappa_3}\right),\nonumber\\
 b_3&=4q_3(d-2),&c_{ij}&=\frac{b_k}{d-2}\quad(\{i,j,k\}=\{1,2,3\}).
\end{align}
Hence
\begin{equation}
 P_3^-\I(\bq+\ee_3)-P_3^+\I(\bq-\ee_3)
 =\omega_d\big[J(1^-)-J(-1^+)\big].
\end{equation}

\emph{Step 3: boundary values.} At $u=-1$ the propagators are regular, by
Eq.~\eqref{eq:Euclidean-propagator}, and the prefactor $w_d(1-u^2)$
vanishes, so that $J(-1^+)=0$. At $u=1$ we first take $0<{\rm Re}(d)<2$. There
$G_q(1)=T_q$ is finite, and the singular branch of each propagator gives the
mass-independent flux
\begin{equation}\label{eq:flux}
 \lim_{u\to1^-}(1-u^2)w_d(u)G_q'(u)=\frac1{\omega_d}.
\end{equation}
In $J$ the term without derivatives vanishes at $u=1$ because of its
prefactor; the terms with two derivatives scale as $(1-u)^{1-d/2}$ and vanish
as well; each term with one derivative tends to $b_iT_jT_k/\omega_d$.
Therefore
\begin{align}
 \omega_d\,J(1^-)&=b_1T_2T_3+b_2T_1T_3+b_3T_1T_2
\cr & =2q_3(d-2)\left[2T_1T_2-(T_1+T_2)T_3
 +\frac{\kappa_1-\kappa_2}{\kappa_3}(T_1-T_2)T_3\right],
\end{align}
which is $\mathcal S_3$, since $d-2=2h-1$. Meromorphic continuation in $d$
establishes Eq.~\eqref{eq:rec} in general.
\end{proof}

To solve the homogeneous equation, define
\begin{equation}\label{eq:H}
 \mathcal H_k(\bq)=\prod_{\boldsymbol\epsilon}
 \frac{\Gamma(B_{\boldsymbol\epsilon})}{\Gamma(A_{\boldsymbol\epsilon})},\quad
 B_{\boldsymbol\epsilon}=\frac{h+q_k+\epsilon_1q_i+\epsilon_2q_j}{2},\quad
 A_{\boldsymbol\epsilon}=B_{\boldsymbol\epsilon}-a.
\end{equation}
The product runs over four sign choices. Since
$\Gamma(B+1/2)/\Gamma(B-1/2)=B-1/2$,
\begin{equation}
 \frac{\mathcal H_k(\bq+\ee_k)}{\mathcal H_k(\bq-\ee_k)}
 =\prod_{\boldsymbol\epsilon}
 \frac{B_{\boldsymbol\epsilon}-1/2}{A_{\boldsymbol\epsilon}-1/2}
 =\frac{P_k^+}{P_k^-}.
\end{equation}
Thus the general homogeneous solution in channel $k$ is $\mathcal H_kp$,
where $p$ has period two in $q_k$.
Let $\bq_n=\bq+2n\ee_k$ and $\widehat\bq_n=\bq+(2n+1)\ee_k$.
Dividing the recurrence by $\mathcal H_k$ gives
\begin{equation}
 \frac{\I(\bq_{n+1})}{\mathcal H_k(\bq_{n+1})}
 -\frac{\I(\bq_n)}{\mathcal H_k(\bq_n)}
 =\frac{\mathcal S_k(\widehat\bq_n)}
 {P_k^+(\widehat\bq_n)\mathcal H_k(\bq_n)}.
\end{equation}
Telescoping yields
\begin{equation}\label{eq:solN}
 \I(\bq)=\mathcal H_k(\bq)
 \frac{\I(\bq_{N+1})}{\mathcal H_k(\bq_{N+1})}
 -\sum_{n=0}^N\frac{\mathcal H_k(\bq)}{\mathcal H_k(\bq+2n\ee_k)}
 \frac{\mathcal S_k(\bq+(2n+1)\ee_k)}{P_k^+(\bq+(2n+1)\ee_k)}.
\end{equation}
The infinite-lattice boundary value is essential. Fix $k=3$, write
$\zeta=q_3$, and set
\begin{equation}\label{wm:eq:bd}
 b_h=\frac{\Gamma(h-1/2)^2\Gamma(3/2-h)\Gamma(2-2h)}
 {(4\pi)^{2h+1}\Gamma(h+1/2)}.
\end{equation}
This is the coefficient of the flat-space sunset with two massless lines
and one line of mass $m$, namely $b_h(m^2)^{d-3}$.
\begin{lemma}
For $2<d<3$ and generic $\zeta$,
\begin{equation}\label{wm:eq:cc}
 \I(1/2,1/2,\zeta)=b_h
 \frac{\Gamma(h+\zeta)\Gamma(h-\zeta)}
 {\Gamma(2-h+\zeta)\Gamma(2-h-\zeta)}.
\end{equation}
\end{lemma}
\begin{proof}
For $t=(1-u)/2$, the conformal propagator follows from
\begin{equation}
 {}_2F_1(h+1/2,h-1/2;d/2;(1+u)/2)=t^{1-d/2}.
\end{equation}
The beta integral gives, initially for ${\rm Re}(\rho)>-1$,
\begin{equation}\label{wm:eq:powerres}
 \int_{S^d}t^\rho G_\zeta(u)\,\dd Y
 =\frac{\Gamma(d/2+\rho)\Gamma(1+\rho)}{\Gamma(d/2)}
 \frac{\Gamma(h+\zeta)\Gamma(h-\zeta)}
 {\Gamma((d+1)/2+\rho+\zeta)\Gamma((d+1)/2+\rho-\zeta)}.
\end{equation}
Substitution of $\rho=2-d$ proves the result.
\end{proof}

\begin{lemma}\label{wm:lem:estimate}
Let $d$ be real, $2<d<3$, and fix nonpolar $q_1,q_2,\zeta$ with
${\rm Im}(\zeta)\ne0$. Then
\begin{equation}\label{wm:eq:estimate}
 \I(q_1,q_2,\zeta+2N)-\I(1/2,1/2,\zeta+2N)
 =O(N^{d-4}\log N).
\end{equation}
\end{lemma}
\begin{proof}
With $c_d=\Gamma(d/2-1)/(4\pi)^{d/2}$, the connection formula gives
\begin{align}
 G_q(1-2t)={}&c_dt^{1-d/2}
 {}_2F_1(1/2-q,1/2+q;2-d/2;t)\nonumber\\
 &+T_q\,{}_2F_1(h+q,h-q;d/2;t).
\end{align}
Consequently,
\begin{align}
 f(t)&=G_{q_1}(1-2t)G_{q_2}(1-2t)-c_d^2t^{2-d}\nonumber\\
 &=c_d(T_{q_1}+T_{q_2})t^{1-d/2}
 +t^{3-d}F_1(t)+t^{2-d/2}F_2(t)+A(t),
\end{align}
where $F_1,F_2,A$ are analytic locally at zero, and $f$ is smooth at one.
The radial spectral decomposition extends to real $d>0$ using Jacobi
polynomials and Gamma-function multiplicities~\cite{Szego1975,DaiXu2013}.
The normalised zonal polynomials are
\begin{equation}
 P_\ell(t)=\frac{\ell!}{(d/2)_\ell}
 P_\ell^{(d/2-1,d/2-1)}(1-2t)
 ={}_2F_1(-\ell,\ell+d-1;d/2;t),\qquad P_\ell(0)=1.
\end{equation}
Their eigenvalues, multiplicities and sphere volume are
\begin{equation}
 \lambda_\ell=\ell(\ell+d-1),\qquad
 D_\ell=\frac{(2\ell+d-1)\Gamma(\ell+d-1)}{\Gamma(d)\Gamma(\ell+1)},
 \qquad V_d=\frac{2\pi^{(d+1)/2}}{\Gamma((d+1)/2)}.
\end{equation}
They satisfy
\begin{equation}
 \langle P_\ell,P_m\rangle=\frac{V_d}{D_\ell}\delta_{\ell m},\qquad
 G_q(1-2t)=\frac1{V_d}\sum_{\ell\ge0}
 \frac{D_\ell P_\ell}{(\ell+h)^2-q^2}.
\end{equation}
For $g=\sum_\ell g_\ell P_\ell$ and $g(t)=t^\rho$, a beta integral and the
terminating Pfaff--Saalschutz identity~\cite{DLMF} give
\begin{equation}\label{wm:eq:powercoeff}
 g_\ell=D_\ell\frac{\Gamma(d/2+\rho)\Gamma(d)}
 {\Gamma(d/2)\Gamma(d+\rho)}
 \frac{(-\rho)_\ell}{(d+\rho)_\ell}
 =O(\ell^{-2\rho-1}).
\end{equation}
Choose a smooth cutoff $\chi$, equal to one near zero and zero near one,
and subtract sufficiently many noninteger Frobenius monomials:
\begin{equation}
 f(t)=\left[c_d(T_{q_1}+T_{q_2})t^{1-d/2}
 +\sum_{j=0}^{L}F_jt^{\rho_j}\right]\chi(t)+R_M(t).
\end{equation}
The sum after the displayed leading term contains the additional monomials. The terms are chosen according to the Frobenius expansion in such the way that near zero $R_M$ is analytic up to terms
$O(t^{\rho_*})$ with $\rho_*>M$. Near one it equals $f$.
Consider the Jacobi operator
\begin{equation}
 J=-\frac1{t^\alpha(1-t)^\alpha}\frac{\dd}{\dd t}
 \left[t^{\alpha+1}(1-t)^{\alpha+1}\frac{\dd}{\dd t}\right],
 \qquad\alpha=d/2-1.
\end{equation}
Its orthogonal
eigenbasis is $P_\ell$, with $JP_\ell=\lambda_\ell P_\ell$.
The operator preserves analytic Taylor expansions and lowers noninteger
powers by at most one. Thus $J^MR_M$ is square integrable, and the boundary
terms in $M$ integrations by parts vanish. For
$R_M=\sum_{m\ge0}r_m P_m$ and $\ell\ge1$, we can write
\begin{equation}
 \langle R_M,P_\ell\rangle=\lambda_\ell^{-M}
 \langle J^MR_M,P_\ell\rangle,
\end{equation}
so the Cauchy--Schwarz inequality gives
\begin{equation}
 |r_\ell|\le C\lambda_\ell^{-M}\|P_\ell\|^{-1}
 =O(\ell^{-2M+(d-1)/2}).
\end{equation}
Here $\|P_\ell\|^2=V_d/D_\ell=O(\ell^{-(d-1)})$.
The next lemma gives $|P_\ell(t)|\le1$ for $d>1$.
The cutoff analytic part is globally smooth and, therefore, has coefficients decaying
faster than any inverse power. A cutoff monomial $\chi t^\rho$ differs
from $t^\rho$ by a smooth function, so Eq.~\eqref{wm:eq:powercoeff}
controls its leading coefficients. The dominant exponent is $1-d/2$;
hence $f_\ell=O((1+\ell)^{d-3})$.

The resolvent pairing is therefore absolutely convergent for $d<3$:
\begin{equation}
 \int_{S^d}f(t)G_w(1-2t)\,\dd Y
 =\sum_{\ell=0}^\infty\frac{f_\ell}{(\ell+h)^2-w^2},
 \qquad w=\zeta+2N.
\end{equation}
Split the sum into $\ell<N$, $N\le\ell\le4N$, and $\ell>4N$.
In the first region both factors of
$(\ell+h-w)(\ell+h+w)$ have magnitude comparable to $N$, so, the sum $S_<$ in the first region satisfies
\begin{equation}
 |S_<|\le\frac C{N^2}\sum_{\ell<N}(1+\ell)^{d-3}
 =O(N^{d-4}).
\end{equation}
In the last region the denominator has magnitude comparable to $\ell^2$, and for the corresponding sum, we get
\begin{equation}|S_>|\le C\sum_{\ell>4N}\ell^{d-5}=O(N^{d-4}).\end{equation}
In the middle region $|\ell+h+w|\sim N$, and $\rm Im\zeta\ne0$ implies
\begin{equation}
 |\ell+h-w|\ge c(1+|\ell-2N|).
\end{equation}
Therefore, the middle sum satisfies
\begin{equation}
 |S_0|\le CN^{d-4}\sum_{|k|\le2N}\frac1{1+|k|}
 =O(N^{d-4}\log N).
\end{equation}
Combining the three estimates proves Eq.~\eqref{wm:eq:estimate}.
\end{proof}
It remains to prove the following Lemma.
\begin{lemma}
For $\lambda=(d-1)/2>0$, it holds $|P_\ell(t)|\le1$ for $0\le t\le1$.
\end{lemma}
\begin{proof}
Set $$t=\frac{1-\cos\theta}{2}.$$
The Poisson representation of the Gegenbauer polynomials gives \cite{DLMF1}
\begin{equation}
 P_{\ell}(t)
 =\frac {\Gamma(\tfrac d2)}{\sqrt \pi \Gamma(h)}\int_0^\pi
 (\cos\theta+i\sin\theta\cos\varphi)^\ell
 (\sin\varphi)^{2\lambda-1}\,\dd\varphi.
\end{equation}
Since
$|\cos\theta+i\sin\theta\cos\varphi|^2
=1-\sin^2\theta\sin^2\varphi\le1$, the triangle inequality proves the bound.
\end{proof}

Gamma-function asymptotics give
$\mathcal H_3(q_1,q_2,\zeta+2N)\sim N^{4a}$. Thus, for $2<d<3$,
\begin{equation}
 \frac{\I(q_1,q_2,\zeta+2N)-\I(1/2,1/2,\zeta+2N)}
 {\mathcal H_3(q_1,q_2,\zeta+2N)}=O(N^{2-d}\log N)\longrightarrow0.
\end{equation}
The reflection formula and asymptotics~\cite{DLMF} yield
\begin{equation}
 \I(1/2,1/2,\zeta+2N)\sim b_hs^+(\zeta)(2N)^{4a},\qquad
 s^+(\zeta)=-\frac{\sin\pi(h+\zeta)}{\sin\pi(h-\zeta)}.
\end{equation}
The physical normalised boundary value is therefore
\begin{equation}\label{wm:eq:boundaryconstant}
 \lim_{N\to\infty}\frac{\I(q_1,q_2,\zeta+2N)}
 {\mathcal H_3(q_1,q_2,\zeta+2N)}=2^{4a}b_hs^+(\zeta).
\end{equation}
It fixes the periodic homogeneous addition along each lattice and gives
\begin{equation}\label{wm:eq:B}
 \Bcal_k(\bq)=2^{4a}b_hs_k^+\mathcal H_k(\bq),\qquad
 s_k^+=-\frac{\sin\pi(h+q_k)}{\sin\pi(h-q_k)}.
\end{equation}
The factor $s_k^+$ encodes the compact-sphere standing-wave behaviour.
Along the lattice $\mu_{q_k}\to-\infty$, so this is not a heavy positive-mass
limit on the sphere. The boundary term is generically nonzero.

\begin{theorem}\label{th:recurrence}
The recurrence and physical normalised boundary value determine
\begin{equation}\label{eq:sol}
 \I=\Bcal_k-\Sigma_d,\qquad
 \Sigma_d=\sum_{n\ge0}\frac{\mathcal H_k(\bq)}{\mathcal H_k(\bq+2n\ee_k)}
 \frac{\mathcal S_k(\bq+(2n+1)\ee_k)}{P_k^+(\bq+(2n+1)\ee_k)},
\end{equation}
where
\begin{equation}\label{eq:B}
 \Bcal_k=2^{4a}b_hs_k^+
 \prod_{\boldsymbol\epsilon}\frac{\Gamma(B_{\boldsymbol\epsilon})}
 {\Gamma(A_{\boldsymbol\epsilon})}.
\end{equation}
\end{theorem}
\begin{remark}
The identity is first obtained for $2<d<3$ and then continued
meromorphically. The three source families have orders
$n^{3-2d}$, $n^{1-d}$ and $n^{-1-d}$, so the sum converges in that initial
interval, away from parameter poles.
\end{remark}
\begin{remark}[The boundary term is not optional]\label{rem:bdry}
The recurrence alone permits additions $\mathcal H_kp$ with $p$ of period
two. Lemma~\ref{wm:lem:estimate} and Eq.~\eqref{wm:eq:boundaryconstant}
fix this ambiguity by comparison with conformally coupled spectators.
Setting $\Bcal_k$ to zero would select the wrong integral.
\end{remark}

For the remainder of this section fix $k=3$.
\begin{theorem}\label{th:solution}
The solution can be written as
\begin{equation}\label{eq:flatgen}\I(q_1,q_2,q_3)=\Bcal_k-\Sigma_d,\end{equation}
with
\begin{equation}\label{eq:flatB}
 \Bcal_k=2^{4a}b_h\left[-\frac{\sin\pi(h+q_3)}{\sin\pi(h-q_3)}\right]
 \prod_{\boldsymbol\epsilon}
 \frac{\Gamma(B_{\boldsymbol\epsilon})}{\Gamma(B_{\boldsymbol\epsilon}-a)},
\end{equation}
\begin{align}\label{eq:flatSigma}
 \Sigma_d=\frac{(q_3+1)(2h-1)}{8\prod_{\boldsymbol\epsilon}B_{\boldsymbol\epsilon}}
 \Big[&2T_{q_1}T_{q_2}\mathcal F_1
 -(T_{q_1}+T_{q_2})T_{q_3+1}\mathcal F_2\nonumber\\
 &+\frac{\kappa_{q_1}-\kappa_{q_2}}{\kappa_{q_3+1}}
 (T_{q_1}-T_{q_2})T_{q_3+1}\mathcal F_3\Big],
\end{align}
where the tadpoles are defined in Eq.~\eqref{eq:tadpole},
and
\begin{align}\label{eq:flatF}
 \mathcal F_1&={}_6F_5\!\left[\begin{matrix}
 B_{\boldsymbol\epsilon}-a,\ (q_3+3)/2,\ 1\\
 B_{\boldsymbol\epsilon}+1,\ (q_3+1)/2
 \end{matrix};1\right],\nonumber\\
 \mathcal F_2&={}_8F_7\!\left[\begin{matrix}
 B_{\boldsymbol\epsilon}-a,\ (q_3+3)/2,\ (q_3+2+a)/2,\ (q_3+3+a)/2,\ 1\\
 B_{\boldsymbol\epsilon}+1,\ (q_3+1)/2,\ (q_3+1-a)/2,\ (q_3+2-a)/2
 \end{matrix};1\right],\\
 \mathcal F_3&={}_8F_7\!\left[\begin{matrix}
 B_{\boldsymbol\epsilon}-a,\ (q_3+3)/2,\ (q_3+1+a)/2,\ (q_3+2+a)/2,\ 1\\
 B_{\boldsymbol\epsilon}+1,\ (q_3+1)/2,\ (q_3+2-a)/2,\ (q_3+3-a)/2
 \end{matrix};1\right].\nonumber
\end{align}
Each occurrence of $B_{\boldsymbol\epsilon}$ stands for its four values,
$B_{\boldsymbol\epsilon}=(h+q_3+\epsilon_1q_1+\epsilon_2q_2)/2$.
\end{theorem}
\begin{proof}
Set $q_{3,n}=q_3+1+2n$ and use $T_{q_{3,n}}$ for the shifted tadpole.
The source is
\begin{equation}\label{Rcal3}
 \mathcal S_3(q_{3,n})=2q_{3,n}(d-2)
 \left[2T_1T_2-(T_1+T_2)T_{q_{3,n}}
 +\frac{\kappa_1-\kappa_2}{\kappa_{q_{3,n}}}
 (T_1-T_2)T_{q_{3,n}}\right].
\end{equation}
Define
\begin{equation}
 {\Theta_n}=\frac{\mathcal H_3(\bq)}{\mathcal H_3(\bq+2n\ee_3)}
 \frac{2q_{3,n}(d-2)}{P_3^+(q_{3,n})}.
\end{equation}
The three terms in $\Sigma_d$ are
\begin{align}
 \Sigma_d^{(1)}&=2T_1T_2\sum_{n\ge0}{\Theta_n},\\
 \Sigma_d^{(2)}&=-(T_1+T_2)\sum_{n\ge0}{\Theta_n}T_{q_{3,n}},\\
 \Sigma_d^{(3)}&=(\kappa_1-\kappa_2)(T_1-T_2)
 \sum_{n\ge0}{\Theta_n}\frac{T_{q_{3,n}}}{\kappa_{q_{3,n}}}.
\end{align}
The Gamma-function shift formula gives
\begin{equation}
 \frac{\mathcal H_3(\bq)}{\mathcal H_3(\bq+2n\ee_3)}
 =\prod_{\boldsymbol\epsilon}\frac{(A_{\boldsymbol\epsilon})_n}
 {(B_{\boldsymbol\epsilon})_n},\qquad
 P_3^+(q_{3,n})=16\prod_{\boldsymbol\epsilon}(n+B_{\boldsymbol\epsilon}).
\end{equation}
Using $(B+1)_n=(n+B)(B)_n/B$ and $e=(q_3+1)/2$,
\begin{equation}
 {\Theta_n}=\frac{4e(d-2)}{16\prod_{\boldsymbol\epsilon}B_{\boldsymbol\epsilon}}
 \frac{(e+1)_n}{(e)_n}
 \prod_{\boldsymbol\epsilon}
 \frac{(A_{\boldsymbol\epsilon})_n}{(B_{\boldsymbol\epsilon}+1)_n}.
\end{equation}
This immediately gives the ${}_6F_5$ in $\Sigma_d^{(1)}$.
For the second term set
\begin{equation}
 \alpha=\frac{h+q_3+1}{2},\qquad\beta=\frac{q_3+1-a}{2}.
\end{equation}
The tadpole recurrence implies
\begin{equation}\label{Tq3n}
 T_{q_{3,n}}=T_{q_3+1}
 \frac{(\alpha)_n(\alpha+1/2)_n}{(\beta)_n(\beta+1/2)_n}.
\end{equation}
Indeed,
\begin{equation}
 \frac{T_{q+2}}{T_q}=\frac{h+q+1}{q+1-a}\frac{h+q}{q-a},
\end{equation}
and multiplication along the lattice proves Eq.~\eqref{Tq3n}.
It gives
\begin{equation}
 \Sigma_d^{(2)}=-(T_1+T_2)T_{q_3+1}
 \frac{4e(d-2)}{16\prod B_{\boldsymbol\epsilon}}
 {}_8F_7\!\left[\begin{matrix}
 A_{\boldsymbol\epsilon},e+1,\alpha,\alpha+1/2,1\\
 B_{\boldsymbol\epsilon}+1,e,\beta,\beta+1/2
 \end{matrix};1\right].
\end{equation}
Finally, with $\gamma=(a+q_3+1)/2$,
\begin{equation}
 \kappa_{q_{3,n}}=-4(\beta+n)(\gamma+n),\qquad
 \kappa_{q_{3,n}}^{-1}=-\frac1{4\gamma\beta}
 \frac{(\gamma)_n(\beta)_n}{(\gamma+1)_n(\beta+1)_n}.
\end{equation}
Multiplication and cancellation of common Pochhammer factors give
\begin{align}
 \Sigma_d^{(3)}=-
 \frac{(\kappa_1-\kappa_2)(T_1-T_2)T_{q_3+1}}{4\gamma\beta}
 \frac{4e(d-2)}{16\prod_{\boldsymbol\epsilon} B_{\boldsymbol\epsilon}}
 {}_8F_7\!\left[\begin{matrix}
 A_{\boldsymbol\epsilon},e+1,\alpha,\gamma,1\\
 B_{\boldsymbol\epsilon}+1,e,\beta+1/2,\beta+1
 \end{matrix};1\right].
\end{align}
These are precisely the three functions in Eq.~\eqref{eq:flatF}.
\end{proof}
\begin{remark}\label{rm-4}
For the expansion near four dimensions it is useful to write
\begin{equation}\label{eq:Sigmahyp}
 \Sigma_d=\frac{4e(d-2)}{16\prod B_{\boldsymbol\epsilon}}
 \left[2T_1T_2\mathcal F_1-(T_1+T_2)T_{q_3+1}\mathcal F_2
 -\frac{(\kappa_1-\kappa_2)(T_1-T_2)T_{q_3+1}}{4\gamma\beta}\mathcal F_3\right].
\end{equation}
The functions $\mathcal F_r$ are regular near $d=4$, away from parameter
poles. The divergences of $\Sigma_d$ are therefore encoded in the tadpole
factors; the boundary term is expanded separately. The summands have orders
$n^{3-2d}$, $n^{1-d}$ and $n^{-1-d}$, respectively. Parameter derivatives
introduce logarithms without spoiling convergence near $d=4$, allowing
termwise expansion.
\end{remark}

\section{The four-dimensional effective potential}\label{sec:5}
To obtain the finite two-loop potential, we need the tadpole Laurent expansion through its coefficient of $\eps$ and the watermelon expansion
through its finite coefficient. Our strategy is the following. We will first expand at fixed labels $\bq$, then insert their regulator dependence at fixed field-dependent masses and scalar
curvature, and finally restore the curvature and renormalisation scales.

\subsection{The four-dimensional expansion}\label{sec:tad}
\subsubsection{The tadpole}
We expand the tadpole of Eq.~\eqref{eq:tadpole} at
$d=4-2\eps$, initially keeping $q$ fixed. Use
\begin{equation}
 \log\Gamma(z+\eps)=\log\Gamma(z)+\psi(z)\eps
 +\frac12\psi'(z)\eps^2+O(\eps^3)
\end{equation}
and
\begin{equation}
 \psi'(\tfrac32+q)+\psi'(\tfrac32-q)
 =-8\frac{1+4q^2}{(1-4q^2)^2}+\frac{\pi^2}{\cos^2\pi q}
\end{equation}
to obtain
\begin{equation}\label{eq:Tlog}
 T_{4-2\eps}(q)=-\frac{1/4-q^2}{16\pi^2\eps}
 \exp\left[-F_1(q)\eps+\frac12F_2(q)\eps^2+O(\eps^3)\right],
\end{equation}
where
\begin{align}
 F_1(q)&=\psi(\tfrac32+q)+\psi(\tfrac32-q)+\gamma_E-1-\log4\pi,\\
 F_2(q)&=1+\frac{\pi^2}{6}-8\frac{1+4q^2}{(1-4q^2)^2}
 +\frac{\pi^2}{\cos^2\pi q}.
\end{align}
Here $\psi'(1)=\sum_{k\ge0}(k+1)^{-2}=\pi^2/6$.
At fixed $q$,
\begin{equation}
 16\pi^2T_d(q)=(\tfrac14-q^2)
 \left[-\frac1\eps+F_1(q)-\frac\eps2(F_1(q)^2+F_2(q))+O(\eps^2)\right].
\end{equation}
At fixed $M_j^2$ and $\Rcal$,
\begin{align}
 q_j(\eps)^2
 &=\frac{(d-1)^2}{4}-d(d-1)\frac{M_j^2}{\Rcal}\nonumber\\
 &=q_j(0)^2-\left(\frac38+\frac76q_j(0)^2\right)\eps
 +\left(\frac14+\frac13q_j(0)^2\right)\eps^2\nonumber\\
 &=\frac94-12\frac{M_j^2}{\Rcal}
 -\left(3-14\frac{M_j^2}{\Rcal}\right)\eps
 +\left(1-4\frac{M_j^2}{\Rcal}\right)\eps^2.
\end{align}
Writing $q=q(0)$ and $s=q^2$, this gives
\begin{align}\label{eq:Tmoving}
 16\pi^2T_d(q(\eps))={}&-\frac{1/4-q^2}{\eps}
 +(\tfrac14-q^2)F_1(q)-\frac38-\frac76q^2\nonumber\\
 &-\left[\frac{1/4-q^2}{2}(F_1(q)^2+F_2(q))
 -\frac14-\frac13q^2\right]\eps\nonumber\\
 &-\left(\frac38+\frac76q^2\right)
 \partial_s[(\tfrac14-q^2)F_1(q)]\eps+O(\eps^2).
\end{align}
The last line contains a moving-label derivative term. Such terms will cancel
from the forest-subtracted expression below. To form tadpole products, write
\begin{equation}
 16\pi^2T_d(q(\eps))=-\frac{Y(q)}\eps+A_0(q)+A_1(q)\eps+O(\eps^2).
\end{equation}
Then
\begin{align}
 (16\pi^2)^2T_{q_i(\eps)}T_{q_j(\eps)}={}&\frac{Y(q_i)Y(q_j)}{\eps^2}
 -\frac{Y(q_i)A_0(q_j)+A_0(q_i)Y(q_j)}\eps\nonumber\\
 &+A_0(q_i)A_0(q_j)-Y(q_i)A_1(q_j)-A_1(q_i)Y(q_j)+O(\eps).
\end{align}

\subsubsection{The watermelon}
The factorisation in Remark~\ref{rm-4} gives, at fixed labels,
\begin{align}
 (16\pi^2)^2\I(\bq)={}&\frac1{\eps^2}
 \left(\frac12\sum_{j=1}^3q_j^2-\frac38\right)\nonumber\\
 &+\frac1\eps\left[
 \sum_{j=1}^3(\tfrac14-q_j^2)(F_1(q_j)-\tfrac12)-\frac12\right]
 +C(\bq)+O(\eps).
\end{align}
Here $C=C_0$ is made explicit in Section~\ref{sec:C}.
Inserting the moving labels gives
\begin{align}
 (16\pi^2)^2\I(\bq(\eps))={}&\frac1{\eps^2}
 \left(\frac12\sum_jq_j^2-\frac38\right)
 +\frac1\eps\left[\sum_j(\tfrac14-q_j^2)(F_1(q_j)+\tfrac1{12})-\frac32\right]
 \nonumber\\
 &+C(\bq)-\frac3{16}-\frac5{12}\sum_jq_j^2
 -\sum_j\left(\frac38+\frac76q_j^2\right)
 \partial_{s_j}[(\tfrac14-q_j^2)F_1(q_j)]+O(\eps),
\end{align}
where $q_j=q_j(0)$ and $s_j=q_j^2$.

\subsubsection{Restoring physical units}
Since $\Rcal=d(d-1)H^2$, at fixed $\Rcal$ and $d=4-2\eps$, $H\equiv H_\eps$ depends on $\eps$.
Therefore, the unit-radius tadpole and watermelon acquire factors
$H_\eps^{2-2\eps}$ and $H_\eps^{2-4\eps}$, respectively, and, at energy scale $\mu$, each loop
contributes $\mu^{2\eps}$. Thus
\begin{align}
 T^{\rm phys}_{q_j(\eps)}
 &=\frac{\Rcal}{12}
 \exp\left[\left(\log\frac{12\mu^2}{\Rcal}+\frac76\right)\eps
 -\frac{59}{72}\eps^2+O(\eps^3)\right]T_{q_j(\eps)},\\
 \I^{\rm phys}(\bq(\eps))
 &=\frac{\Rcal}{12}
 \exp\left[\left(2\log\frac{12\mu^2}{\Rcal}+\frac76\right)\eps
 -\frac{143}{72}\eps^2+O(\eps^3)\right]\I(\bq(\eps)).
\end{align}
For clarity, define
\begin{equation}
 Y_j=\frac14-q_j^2,\quad {\sigma_j}=\frac38+\frac76q_j^2,\quad
 {\tau_j}=\frac14+\frac13q_j^2,\quad L_\mu=\log\frac{12\mu^2}{\Rcal},
 \quad {\Phi_j}=\partial_{s_j}[Y_jF_1(q_j)].
\end{equation}
These are local abbreviations for the following expansions:
\begin{align}\label{eq:Tphysical-expanded}
 16\pi^2\frac{12}{\Rcal}T^{\rm phys}_{q_j(\eps)}
 ={}&-\frac{Y_j}\eps+Y_jF_1(q_j)-\frac23-Y_jL_\mu\nonumber\\
 &+\eps\Bigg[-\frac{Y_j}{2}(F_1(q_j)^2+F_2(q_j))+{\tau_j}
 +(Y_jF_1(q_j)-{\sigma_j})(L_\mu+\tfrac76)\nonumber\\
 &\hspace{1cm}-Y_j\left(\frac12(L_\mu+\tfrac76)^2-\frac{59}{72}\right)
 -{\sigma_j}{\Phi_j}\Bigg]+O(\eps^2).
\end{align}
{The moving-label Laurent coefficients of the watermelon are}
\begin{align}
 {\hat C_{-2}}=&\frac12\sum_jq_j^2-\frac38,\qquad
 {\hat C_{-1}}=\sum_jY_j(F_1(q_j)+\tfrac1{12})-\frac32,\cr
 {\hat C_0}=&C(\bq)-\frac3{16}-\frac5{12}\sum_jq_j^2-\sum_j{\sigma_j}{\Phi_j}.
\end{align}
The physical watermelon expansion is
\begin{align}\label{eq:Iphysical-expanded}
 (16\pi^2)^2\frac{12}{\Rcal}\I^{\rm phys}(\bq(\eps))
 ={}&\frac{{\hat C_{-2}}}{\eps^2}+\frac{{\hat C_{-1}}+(2L_\mu+7/6){\hat C_{-2}}}\eps+{\hat C_0}
 +(2L_\mu+\tfrac76){\hat C_{-1}}\nonumber\\
 &+\left[\frac12(2L_\mu+\tfrac76)^2-\frac{143}{72}\right]{\hat C_{-2}}+O(\eps).
\end{align}
These compact forms retain all terms required for the forest subtraction.

\subsection{Renormalisation}
We use dimensional regularisation and minimal subtraction, organised by the
forest formula. We first subtract the pole parts of divergent proper
subgraphs and then subtract the remaining local poles. The renormalised
tadpole is
\begin{equation}\label{tadpole-finito}
 16\pi^2\frac{12}{\Rcal}t(M_j)
 =(\tfrac14-q_j^2)F_1(q_j)-\frac23
 -(\tfrac14-q_j^2)\log\frac{12\mu^2}{\Rcal}.
\end{equation}
Its higher-order terms remain those in Eq.~\eqref{eq:Tphysical-expanded}.
The watermelon has three divergent bubble subgraphs, one for each pair of
lines. Each bubble contributes a pole $1/(16\pi^2\eps)$ multiplying the
tadpole of the complementary line~\cite{CEM2024}. The first forest step is
\begin{equation}
 \I^*(\bq(\eps))=\I^{\rm phys}(\bq(\eps))
 -\frac1{16\pi^2\eps}\sum_{j=1}^3T^{\rm phys}_{q_j(\eps)}.
\end{equation}
Substitution gives
\begin{align}\label{pb:forest}
 (16\pi^2)^2\frac{12}{\Rcal}\I^*(\bq(\eps))={}&
 -\frac1{\eps^2}\left(\frac12\sum_jq_j^2-\frac38\right)
 +\frac1\eps\left(\frac12\sum_jq_j^2+\frac18\right)\nonumber\\
 &+C(\bq)+\frac12\sum_j(\tfrac14-q_j^2)
 [F_1(q_j)^2+F_2(q_j)]\nonumber\\
 &+L_\mu\left[\sum_j(\tfrac14-q_j^2)F_1(q_j)
 +\sum_jq_j^2-\frac74\right]\nonumber\\
 &+L_\mu^2\left(\frac12\sum_jq_j^2-\frac38\right)
 -\frac{11}{12}+O(\eps).
\end{align}
The moving-label derivative terms cancel. The remaining poles are local
and are removed by the two-loop counterterms. The fully renormalised
watermelon is
\begin{align}\label{eq:finiteI}
 \mathcal W(M_1,M_2,M_3)=\frac{\Rcal}{12(16\pi^2)^2}\Bigg\{&
 C(\bq)+\frac12\sum_j(\tfrac14-q_j^2)[F_1(q_j)^2+F_2(q_j)]\nonumber\\
 &+L_\mu\left[\sum_j(\tfrac14-q_j^2)F_1(q_j)
 +\sum_jq_j^2-\frac74\right]\nonumber\\
 &+L_\mu^2\left(\frac12\sum_jq_j^2-\frac38\right)-\frac{11}{12}\Bigg\}.
\end{align}
Equivalently,
\begin{align}
 \I^{\rm phys}(\bq(\eps))={}&\frac{\Rcal}{12(16\pi^2)^2}
 \left[-\frac{\frac12\sum_jq_j^2-3/8}{\eps^2}
 +\frac{\frac12\sum_jq_j^2+1/8}{\eps}\right]\nonumber\\
 &+\mathcal W(M_1,M_2,M_3)
 +\frac1{16\pi^2\eps}\sum_jT^{\rm phys}_{q_j(\eps)}+O(\eps).
\end{align}
The cancellation concerns the displayed moving-label derivative terms after
subdivergence subtraction; it does not assert that every intermediate
Laurent coefficient is independent of the continuation prescription.

\subsection{The effective potential}
Equation~\eqref{tadpole-finito} shows that the fixed-mass, fixed-curvature
prescription adds a constant $-2/3$ to the dimensionless finite tadpole.
Its contribution to the integrated one-loop potential is $q^2/3$.
The renormalised contribution of a single mode is
\begin{align}\label{eq:V1}
 \mathcal V_1(M)={}&\frac{\Rcal^2}{144(16\pi^2)}
 \left[\tilde B(q)+\frac14(\tfrac14-q^2)^2
 \left(\psi(\tfrac32+q)+\psi(\tfrac32-q)-2
 -\log\frac{12Q^2}{\Rcal}\right)+\frac{q^2}{3}\right]\nonumber\\
 ={}&\frac{\Rcal}{48}\left[(\tfrac14-q^2)t(M)
 +\frac{\Rcal}{48\pi^2}
 \left(\tilde B(q)-\frac{q^4}{4}+\frac{7q^2}{24}+\frac5{192}\right)\right],
\end{align}
where
\begin{align}\label{Btildefunction}
 \tilde B(q)=B(-iq)
 &=\frac{(1/4-q^2)^2}{4}
 -\int^q\frac14(\tfrac14-y^2)^2
 [\psi'(-\tfrac12+y)-\psi'(-\tfrac12-y)]\,\dd y\nonumber\\
 &=-\frac12\int^{1/4-q^2}\!\dd y\,y^2
 \left[\sum_{n=0}^\infty\frac{2n-1}{2(n(n-1)+y)^2}-\frac1y\right],\\
 q^2&=\frac94-\frac{12M^2}{\Rcal},\qquad
 Q^2=4\pi\mu^2e^{-\gamma_E}.
\end{align}
The subscript $r$ below denotes renormalised quantities. The masses are
\begin{equation}
 M_0^2=m_r^2+\xi_r\Rcal+3c_r\varphi_r^2,\qquad
 M_1^2=m_r^2+\xi_r\Rcal+c_r\varphi_r^2.
\end{equation}
The mass integration leaves a field-independent curvature normalisation.
Only local ambiguities of the form allowed by the gravitational action can
be absorbed into its finite couplings for varying $\Rcal$. The scalar
potential below is specified up to this gravitational normalisation; we do
not infer an arbitrary curvature function from the tadpole equation.

The double-tadpole contractions have weights $3$, $2(N-1)$ and $N^2-1$
for masses $(M_0,M_0)$, $(M_0,M_1)$ and $(M_1,M_1)$, respectively~\cite{ACM}.
Thus
\begin{equation}
 V_T(\varphi_r)=\frac{c_r}{4}
 [3t(M_0)^2+2(N-1)t(M_0)t(M_1)+(N^2-1)t(M_1)^2].
\end{equation}
The cubic vertices give six contractions for three radial lines and
$2(N-1)$ for one radial and two transverse lines. Including the expansion
factor $1/2$, their contribution is
\begin{equation}
 V_W(\varphi_r)=-c_r^2\varphi_r^2
 [3\mathcal W(M_0,M_0,M_0)+(N-1)\mathcal W(M_1,M_1,M_0)].
\end{equation}
Restoring $\hbar$, the final potential is
\begin{align}\label{eq:Vfull}
 V_{\rm eff}(\varphi_r)={}&\mathcal L_{g,r}
 +\frac12(m_r^2+\xi_r\Rcal)\varphi_r^2
 +\frac{c_r}{4}\varphi_r^4
 +\hbar\big[\mathcal V_1(M_0)+(N-1)\mathcal V_1(M_1)\big]
 \nonumber\\
 &+\hbar^2\frac{c_r}{4}\big[
 3t(M_0)^2+2(N-1)t(M_0)t(M_1)+(N^2-1)t(M_1)^2\big]
 \nonumber\\
 &-\hbar^2c_r^2\varphi_r^2\big[
 3\mathcal W(M_0,M_0,M_0)+(N-1)\mathcal W(M_1,M_1,M_0)\big]
 +O(\hbar^3).
\end{align}

For explicit evaluation define
\begin{align}
 F(M)&=\psi(\tfrac32+q(M))+\psi(\tfrac32-q(M))-1,\\
 G(M)&=1+\frac{\pi^2}{6}
 +\frac16\left(\frac{M^2}{\Rcal}-\frac16\right)^{-1}
 -\frac1{144}\left(\frac{M^2}{\Rcal}-\frac16\right)^{-2}
 +\frac{\pi^2}{\cos^2(\pi q(M))},\\
 q(M)^2&=\frac94-\frac{12M^2}{\Rcal}.
\end{align}
The tadpole can then be written compactly as
\begin{equation}\label{eq:t-final}
 t(M)=\frac{\Rcal}{16\pi^2}
 \left[\left(\frac{M^2}{\Rcal}-\frac16\right)
 \left(F(M)-\log\frac{12Q^2}{\Rcal}\right)-\frac1{18}\right].
\end{equation}
The finite-part conversion is
\begin{equation}\label{pb:CQ}
 C_Q=C+2\lambda_E\sum_j(\tfrac14-q_j^2)F(M_j)
 +(\lambda_E^2-\lambda_E)\sum_j(\tfrac14-q_j^2)-\lambda_E,
 \qquad\lambda_E=\gamma_E-\log4\pi.
\end{equation}
For any three line masses, set
\begin{equation}
 x_a=\frac{M_a^2}{\Rcal}-\frac16,\qquad
 L_Q=\log\frac{12Q^2}{\Rcal},\qquad
 D_{ijk}=C_Q(q(M_i),q(M_j),q(M_k)).
\end{equation}
Then Eq.~\eqref{eq:finiteI} becomes
\begin{align}\label{eq:W-final}
 \mathcal W(M_i,M_j,M_k)=&\frac{\Rcal}{(16\pi^2)^2}\Bigg\{
 \frac{D_{ijk}-11/12}{12}
 +\frac12\sum_{a\in(i,j,k)}x_a[F(M_a)^2+G(M_a)]\nonumber\\
 &+L_Q\left[\sum_{a\in(i,j,k)}x_aF(M_a)
 -\sum_{a\in(i,j,k)}\frac{M_a^2}{\Rcal}+\frac5{12}\right]
 -\frac{L_Q^2}{2}\sum_{a\in(i,j,k)}x_a\Bigg\}.
\end{align}
The sums run over the three propagator lines, counting repeated masses
with multiplicity. In particular,
\begin{align}
 \mathcal W(M_1,M_1,M_0)={}&\frac{\Rcal}{(16\pi^2)^2}\Bigg\{
 \frac{D_{110}-11/12}{12}+x_1[F(M_1)^2+G(M_1)]\nonumber\\
 &\quad+\frac{x_0}{2}[F(M_0)^2+G(M_0)]\nonumber\\
 &\quad+L_Q\left[2x_1F(M_1)+x_0F(M_0)
 -\frac{2M_1^2+M_0^2}{\Rcal}+\frac5{12}\right]\nonumber\\
 &\quad-L_Q^2\left(\frac{M_1^2}{\Rcal}+\frac{M_0^2}{2\Rcal}-\frac14\right)\Bigg\},\\
 \mathcal W(M_0,M_0,M_0)={}&\frac{\Rcal}{(16\pi^2)^2}\Bigg\{
 \frac{D_{000}-11/12}{12}+\frac{3x_0}{2}[F(M_0)^2+G(M_0)]\nonumber\\
 &\quad+L_Q\left[3x_0F(M_0)-\frac{3M_0^2}{\Rcal}+\frac5{12}\right]\nonumber\\
 &\quad-L_Q^2\left(\frac{3M_0^2}{2\Rcal}-\frac14\right)\Bigg\}.
\end{align}

\subsection{The finite-part function \texorpdfstring{$C(\bq)$}{C(q)}}\label{sec:C}
By definition, $(16\pi^2)^{-2}C(\bq)$ is the finite part of
Eq.~\eqref{eq:flatgen}. All expansions in this subsection are at fixed
labels. Each hypergeometric parameter is either constant in $\eps$ or
affine with slope $\pm1/2$. Away from parameter poles, the functions
$\mathcal F_r$ are analytic near $\eps=0$.
If the coefficient of $\mathcal F_r$ in Eq.~\eqref{eq:flatSigma}, multiplied
by $(16\pi^2)^2$, is
\begin{equation}\label{eq:plaurent}
 \frac{p_{-2}^{(r)}}{\eps^2}+\frac{p_{-1}^{(r)}}\eps+p_0^{(r)}+O(\eps),
\end{equation}
then
\begin{equation}\label{eq:C0deriv}
 C_0=\mathcal B_0-\sum_{r=1}^3\left[
 \frac{p_{-2}^{(r)}}2\frac{\dd^2\mathcal F_r}{\dd\eps^2}\bigg|_0
 +p_{-1}^{(r)}\frac{\dd\mathcal F_r}{\dd\eps}\bigg|_0
 +p_0^{(r)}\mathcal F_r(0)\right].
\end{equation}
The coefficients are built from Gamma functions and rational factors.
For a compact explicit expression, define the following four pairs at $d=4$:
\begin{equation}\label{eq:AB-four}
 A_{st}^{(4)}=\frac{1/2+q_3+sq_1+tq_2}{2},\qquad
 B_{st}^{(4)}=\frac{3/2+q_3+sq_1+tq_2}{2},\qquad s,t=\pm1,
\end{equation}
and
\begin{equation}\label{eq:Pi}
 \Pi=\prod_{s,t=\pm1}\frac{\Gamma(B_{st}^{(4)})}{\Gamma(A_{st}^{(4)})},
 \qquad P_0=\prod_{s,t=\pm1}(\tfrac32+q_3+sq_1+tq_2).
\end{equation}
\begin{theorem}
The finite part is
\begin{align}\label{eq:C0replacement-dderivatives}
 C=\mathcal B_0-\frac{4(q_3+1)}{P_0}\sum_{r=1}^3\Bigg[
 &2u_r\partial_d^2\mathcal F_r(4)
 -2(w_r+\ell u_r)\partial_d\mathcal F_r(4)\nonumber\\
 &+\left(r_r+\ell w_r+\frac{\ell^2+m}{2}u_r\right)\mathcal F_r(4)\Bigg],
\end{align}
where
\begin{align}
 \mathcal B_0&=2\left(\frac{{X_{\mathcal B}}^2}{2}+{Y_{\mathcal B}}\right)\Pi,\\
 {X_{\mathcal B}}&=2\log\pi+3-2\gamma_E+2\pi\tan(\pi q_3)
 -\frac12\sum_{s,t}[\psi(A_{st}^{(4)})+\psi(B_{st}^{(4)})],\\
 {Y_{\mathcal B}}&=\frac{\pi^2}{2}+\frac52
 +\frac18\sum_{s,t}[\psi'(B_{st}^{(4)})-\psi'(A_{st}^{(4)})],\\
 \ell&=-1+\sum_{s,t}\frac1{3/2+q_3+sq_1+tq_2},\qquad
 m=-1+\sum_{s,t}\frac1{(3/2+q_3+sq_1+tq_2)^2}.
\end{align}
Writing $Y_q=1/4-q^2$, the remaining coefficients are
\begin{align}
 u_1&=2Y_{q_1}Y_{q_2},\\
 u_2&=-(Y_{q_1}+Y_{q_2})Y_{q_3+1},\\
 u_3&=(q_2^2-q_1^2)^2,\\
 w_1&=-2[Y_{q_1}a_{q_2}+Y_{q_2}a_{q_1}],\\
 w_2&=(Y_{q_1}+Y_{q_2})a_{q_3+1}
 +Y_{q_3+1}(a_{q_1}+a_{q_2}),\\
 w_3&=(q_2^2-q_1^2)^2\chi_3
 -(q_2^2-q_1^2)(a_{q_1}-a_{q_2}),\\
 r_1&=2[a_{q_1}a_{q_2}-Y_{q_1}b_{q_2}-Y_{q_2}b_{q_1}],\\
 r_2&=(Y_{q_1}+Y_{q_2})b_{q_3+1}
 +Y_{q_3+1}(b_{q_1}+b_{q_2})-(a_{q_1}+a_{q_2})a_{q_3+1},\\
 r_3&=(q_2^2-q_1^2)^2
 \left(\frac{\chi_3^2}{2}+\frac12+\frac{\pi^2}{12}
 +\frac{\pi^2}{2\cos^2\pi q_3}\right)\nonumber\\
 &\quad-(q_2^2-q_1^2)(a_{q_1}-a_{q_2})\chi_3
 -(q_2^2-q_1^2)(b_{q_1}-b_{q_2}),
\end{align}
with
\begin{align}
 \chi_3&=\log4\pi+1-\gamma_E-\psi(\tfrac32+q_3)-\psi(-\tfrac12-q_3),\\
 a_q&=1-Y_q[\log4\pi+1-\gamma_E-\psi(\tfrac12+q)-\psi(\tfrac12-q)],\\
 b_q&=\log4\pi-\gamma_E-\psi(\tfrac12+q)-\psi(\tfrac12-q)\nonumber\\
 &\quad-Y_q\left\{\frac12[\log4\pi+1-\gamma_E
 -\psi(\tfrac12+q)-\psi(\tfrac12-q)]^2
 +\frac12+\frac{\pi^2}{12}+\frac{\pi^2}{2\cos^2\pi q}\right\}.
\end{align}
\end{theorem}
\begin{proof}
Expand Eq.~\eqref{eq:flatgen} directly in $\eps$, using
$\partial_\eps=-2\partial_d$. The tadpole products give $u_r,w_r,r_r$;
the rational prefactor gives $\ell,m$, and the boundary term gives
$\mathcal B_0$.
\end{proof}
The dimension derivatives in this formula can be evaluated by convergent
series. At $d=4$ the undifferentiated constants themselves admit closed
expressions in Gamma functions and rational functions.

\begin{theorem}[Evaluation at $d=4$]\label{thm:closed}
With $\Pi,P_0$ from Eq.~\eqref{eq:Pi},
\begin{align}\label{eq:F1closed}
 \mathcal F_1(4)
 &=\frac{P_0}{16(q_3+1)Y_{q_1}Y_{q_2}}
 \left[\frac14+q_3^2-q_1^2-q_2^2-4\Pi\right],\\
 \mathcal F_2(4)
 &=\frac{P_0}{16(q_3+1)Y_{q_1}Y_{q_2}Y_{q_3+1}}
 \left[(q_2^2-q_1^2)^2
 -(Y_{q_1}+Y_{q_2})Y_{q_3}
 -4(Y_{q_1}+Y_{q_2})\Pi\right],\label{eq:F2closed}\\
 \mathcal F_3(4)&=\mathcal F_1(4).\label{eq:F3collapse}
\end{align}
\end{theorem}
\begin{proof}
In $\mathcal F_3(4)$, two upper parameters cancel two lower parameters,
leaving $\mathcal F_1(4)$ before summation.
For the other two identities, put
\begin{equation}
 S=Y_{q_1}+Y_{q_2},\qquad D=Y_{q_1}-Y_{q_2},\qquad P=Y_{q_1}Y_{q_2},
\end{equation}
and, with $q_{3,n}=q_3+1+2n$, define
\begin{equation}\label{eq:K}
 W_n=\prod_{s,t}\frac{(A_{st}^{(4)})_n}{(B_{st}^{(4)})_n},\qquad
 K_n=\frac{4q_{3,n}W_n}
 {\prod_{s,t}(1/2+q_{3,n}+sq_1+tq_2)}.
\end{equation}
Two rational identities telescope the sums:
\begin{align}\label{eq:tel}
 W_{n+1}\frac{Y_{q_{3,n}+1}-S}{4P}
 -W_n\frac{Y_{q_{3,n}-1}-S}{4P}&=K_n,\nonumber\\
 W_{n+1}\frac{SY_{q_{3,n}+1}-D^2}{4P}
 -W_n\frac{SY_{q_{3,n}-1}-D^2}{4P}&=K_nY_{q_{3,n}}.
\end{align}
The upper endpoint terms tend to $-\Pi/P$ and $-S\Pi/P$, respectively.
Hence
\begin{equation}\label{eq:sums}
 \sum_{n\ge0}K_n=\frac{S-Y_{q_3}-4\Pi}{4P},\qquad
 \sum_{n\ge0}K_nY_{q_{3,n}}=\frac{D^2-SY_{q_3}-4S\Pi}{4P}.
\end{equation}
The ${}_6F_5$ summand is $K_n/K_0$, where $K_0=4(q_3+1)/P_0$.
The ${}_8F_7$ summand at $d=4$ is
$(K_n/K_0)(Y_{q_{3,n}}/Y_{q_3+1})$. These give the stated evaluations.
\end{proof}

\paragraph{An integral-free convergent series.}
Introduce
\begin{align}
 Y_q=&\frac14-q^2,\qquad
 g_q=\log4\pi+1-\gamma_E-\psi(\tfrac12+q)-\psi(\tfrac12-q),\cr
 \Omega_q=&\frac12+\frac{\pi^2}{12}+\frac{\pi^2}{2\cos^2\pi q}.
\end{align}
Using Eq.~\eqref{eq:AB-four}, define the weight and its logarithmic variations:
\begin{align}
 K_n&=\frac{4(q_3+1+2n)}{\prod_{s,t}(3/2+q_3+sq_1+tq_2+2n)}
 \prod_{s,t}\frac{(A_{st}^{(4)})_n}{(B_{st}^{(4)})_n},\\
 L_n&=-1+\frac12\sum_{s,t}\big[
 \psi(n+A_{st}^{(4)})-\psi(A_{st}^{(4)})
 +\psi(n+B_{st}^{(4)}+1)-\psi(B_{st}^{(4)})\big],\\
 {\Lambda_n}&=-1+\frac14\sum_{s,t}\big[
 \psi'(n+A_{st}^{(4)})-\psi'(A_{st}^{(4)})
 -\psi'(n+B_{st}^{(4)}+1)+\psi'(B_{st}^{(4)})\big].
\end{align}
\begin{theorem}
With $q_{3,n}=q_3+1+2n$, the finite part is
\begin{align}\label{eq:C0replacement-shortseries}
 C_0=\mathcal B_0-\sum_{n=0}^\infty K_n\Big[&
 Y_{q_1}Y_{q_2}Z_n(q_1,q_2)
 +\frac12Y_{q_1}(q_2^2-q_1^2-Y_{q_{3,n}})Z_n(q_1,q_{3,n})\nonumber\\
 &-\frac12Y_{q_2}(q_2^2-q_1^2+Y_{q_{3,n}})Z_n(q_2,q_{3,n})
 -2(Y_{q_1}+Y_{q_2})J_n(q_1,q_2)\nonumber\\
 &+(Y_{q_2}+Y_{q_{3,n}})J_n(q_1,q_{3,n})
 +(Y_{q_1}+Y_{q_{3,n}})J_n(q_2,q_{3,n})\nonumber\\
 &+Y_{q_1}+Y_{q_2}-2Y_{q_{3,n}}\Big],
\end{align}
where
\begin{equation}
 J_n({q,q'})=L_n+g_{{q}}+g_{{q'}},\qquad
 Z_n({q,q'})=J_n({q,q'})^2+{\Lambda_n}+2\Omega_{{q}}+2\Omega_{{q'}}.
\end{equation}
\end{theorem}
\begin{proof}
For $d=4-2\eps$, $\kappa_q=Y_q-\eps+\eps^2$, and
\begin{align}
 16\pi^2T_q&=-\frac{\kappa_q}\eps E_q(\eps),\\
 E_q(\eps)&=\frac{\Gamma(1+\eps)}{1-\eps}(4\pi)^\eps
 \frac{\Gamma(1/2+q-\eps)\Gamma(1/2-q-\eps)}
 {\Gamma(1/2+q)\Gamma(1/2-q)}
 =\exp[\eps g_q+\eps^2\Omega_q+O(\eps^3)],\\
 \frac{T_q}\eta&=-\frac{Y_q}\eps+a_q+b_q\eps+O(\eps^2),\qquad
 \eta=\frac1{16\pi^2}.
\end{align}
After multiplying by $\eps^2$, the normalised source bracket is
\begin{align}
 &2\kappa_{q_1}\kappa_{q_2}E_{q_1}E_{q_2}
 +\kappa_{q_1}(q_2^2-q_1^2-\kappa_{q_{3,n}})E_{q_1}E_{q_{3,n}}-\kappa_{q_2}(q_2^2-q_1^2+\kappa_{q_{3,n}})E_{q_2}E_{q_{3,n}}.
\end{align}
For $v=q_{3,n}$, the polynomial coefficients are
\begin{equation}
\begin{array}{c|ccc}
 &1&\eps&\eps^2\\\hline
2\kappa_{q_1}\kappa_{q_2}&2Y_{q_1}Y_{q_2}&-2S&2(S+1)\\
\kappa_{q_1}(q_2^2-q_1^2-\kappa_v)
 &Y_{q_1}(q_2^2-q_1^2-Y_v)&Y_{q_2}+Y_v&-(Y_{q_2}+Y_v+1)\\
-\kappa_{q_2}(q_2^2-q_1^2+\kappa_v)
 &-Y_{q_2}(q_2^2-q_1^2+Y_v)&Y_{q_1}+Y_v&-(Y_{q_1}+Y_v+1)
\end{array}
\end{equation}
where $S=Y_{q_1}+Y_{q_2}$. Multiplying by {the corresponding products $E_qE_{q'}$} and
\begin{equation}
 \frac{2q_{3,n}(d-2)W_n(d)}{P_3^+(q_{3,n};d)}
 =K_n\exp[\eps L_n+\tfrac12\eps^2{\Lambda_n}+O(\eps^3)]
\end{equation}
gives Eq.~\eqref{eq:C0replacement-shortseries}.
\end{proof}
This representation uses Gamma functions, digammas, trigammas and one
absolutely convergent series. Appendix~\ref{sec:abel} gives a form whose
summand contains only linear digammas.
\begin{remark}[Convergence and interchange]
On compact sets avoiding poles of the split representation, for large $n$, we have $K_n=O(n^{-5})$, $Y_{q_{3,n}}=O(n^2)$, $g_{q_{3,n}}=O(\log n)$,
$L_n=O(\log n)$ and ${\Lambda_n}=O(1)$. The finite-part summand is therefore
$O(n^{-3}\log^2n)$ and absolutely summable.
After removing its explicit double pole, the generic-dimension summand is
holomorphic in $\eps$. For $|{\rm Re}(\eps)|\le\delta<1/2$, its two slow families
and their first two derivatives are dominated by
$n^{-5+4\delta}(1+\log^2n)$ and $n^{-3+2\delta}(1+\log^2n)$.
The Weierstrass theorem permits termwise differentiation and extraction of
$C_0$. This argument excludes crossing an unremoved polar divisor, but does
not require unequal or ordered masses.
\end{remark}

\subsubsection{Two equal and three equal labels}
For $(q_1,q_2,q_3)=(q_1,q_1,q_0)$, take the repeated labels as spectators.
The difference $q_2^2-q_1^2$ vanishes identically in every dimension, so the
third hypergeometric term drops out. With $q_{0,n}=q_0+1+2n$,
\begin{align}\label{eq:C0replacement-twoequal}
 C_0(q_1,q_1,q_0)=\mathcal B_0(q_1,q_1,q_0)-\sum_{n\ge0}K_n\Big[&
 Y_{q_1}^2\{(L_n+2g_{q_1})^2+{\Lambda_n}+4\Omega_{q_1}\}\nonumber\\
 &-Y_{q_1}Y_{q_{0,n}}
 \{(L_n+g_{q_1}+g_{q_{0,n}})^2+{\Lambda_n}+2\Omega_{q_1}+2\Omega_{q_0}\}\nonumber\\
 &-4Y_{q_1}(L_n+2g_{q_1})
 +2(Y_{q_1}+Y_{q_{0,n}})(L_n+g_{q_1}+g_{q_{0,n}})\nonumber\\
 &+2Y_{q_1}-2Y_{q_{0,n}}\Big].
\end{align}
Here $\Omega_{q_{0,n}}=\Omega_{q_0}$ because their arguments differ by an
integer. A further reduction holds before taking $d\to4$. Set
\begin{equation}
 A_\pm=\frac{2-h+p\pm2q}{2},\quad
 B_\pm=\frac{h+p\pm2q}{2},\quad
 \beta=\frac{p+2-h}{2},\quad\gamma=\frac{p+h}{2},\quad e=\frac{p+1}{2}.
\end{equation}
The two nonzero hypergeometric constants are both ${}_6F_5$:
\begin{align}
 \mathcal F_1(d;q,q,p)&={}_6F_5\!\left[\begin{matrix}
 A_+,\beta,\beta,A_-,e+1,1\\
 B_++1,\gamma+1,\gamma+1,B_-+1,e
 \end{matrix};1\right],\\
 \mathcal F_2(d;q,q,p)&={}_6F_5\!\left[\begin{matrix}
 A_+,\beta,A_-,e+1,\gamma+1/2,1\\
 B_++1,\gamma+1,B_-+1,e,\beta+1/2
 \end{matrix};1\right].
\end{align}
These identities hold in $d$, so dimension derivatives must be total
derivatives along all the displayed $d$-dependent parameters.
For three equal labels, set $q_1=q_0$ in
Eq.~\eqref{eq:C0replacement-twoequal}.

\section{Beta functions}\label{sec:6}
Finally, we derive the beta functions up to two loops.
\subsection{Divergences and beta functions in four dimensions}
Set $v=\varphi_r$, $n=N-1$, $K=16\pi^2$, and suppress the subscript $r$
on the couplings. Define
\begin{equation}
 A=m^2+(\xi-1/6)\Rcal,\qquad
 \Delta_0=A+3cv^2,\qquad\Delta_1=A+cv^2.
\end{equation}
The matter-dependent one-loop counterterm is
\begin{equation}\label{pb:P1}
 \delta V_1=\frac{\Delta_0^2+n\Delta_1^2}{4K\eps},\qquad
 \Delta_0^2+n\Delta_1^2=NA^2+2(N+2)cAv^2+(N+8)c^2v^4.
\end{equation}
The geometric heat-kernel term adds $-N\Rcal^2/(4320K\eps)$ on $S^4$.
The one-loop Hessian insertions are
\begin{align}
 \delta M_{0,1}^2&=\frac{(N+2)cA+3(N+8)c^2v^2}{K\eps},\\
 \delta M_{1,1}^2&=\frac{(N+2)cA+(N+8)c^2v^2}{K\eps}.
\end{align}
From Eq.~\eqref{pb:forest}, the subdivergence-subtracted sunset is
\begin{equation}
 R'\mathcal I_{ijk}=\frac1{K^2}\left[
 \frac{\Delta_i+\Delta_j+\Delta_k}{2\eps^2}
 +\frac{\Rcal/24-(\Delta_i+\Delta_j+\Delta_k)/2}{\eps}\right]
 +\mathcal W_{ijk}+O(\eps).
\end{equation}
Combining the two vacuum topologies with
$\frac12(\delta M_{0,1}^2T_0+n\delta M_{1,1}^2T_1)$ cancels every
nonlocal simple pole. The two-loop action counterterm is
\begin{align}\label{pb:P2}
 \delta V_2={}&\frac1{4K^2\eps^2}
 [N(N+2)cA^2+2(N+2)(N+5)c^2Av^2+(N+8)^2c^3v^4]\nonumber\\
 &-\frac{c^2v^2}{4K^2\eps}
 \left[6(N+2)A+2(5N+22)cv^2-\frac{(N+2)\Rcal}{6}\right].
\end{align}
Here $\delta V_1+\delta V_2$ denotes action counterterms expressed in the
renormalised field. These differ from bare-parameter counterterms once
field normalisation runs. Writing
\begin{equation}
 \delta V=\delta\rho+\delta\kappa\Rcal+\delta b\Rcal^2
 +\frac12(\delta m_A^2+\delta\xi_A\Rcal)v^2+\frac14\delta c_Av^4,
\end{equation}
the defining relations are
\begin{equation}
 \phi_B=Z_\phi^{1/2}\phi,\quad
 m_B^2=\frac{m^2+\delta m_A^2}{Z_\phi},\quad
 \xi_B=\frac{\xi+\delta\xi_A}{Z_\phi},\quad
 c_B=\mu^{2\eps}\frac{c+\delta c_A}{Z_\phi^2}.
\end{equation}
The two-loop simple residues in Eq.~\eqref{pb:P2} are
\begin{align}
 \delta m_{A,21}^2&=-\frac{3(N+2)c^2m^2}{K^2},\\
 \delta\xi_{A,21}&=\frac{(N+2)c^2}{K^2}
 \left[-3(\xi-1/6)+\frac1{12}\right],\\
 \delta c_{A,21}&=-\frac{2(5N+22)c^3}{K^2}.
\end{align}
The index $21$ denotes the coefficient of $1/\eps$ at two loops.
The mass and curvature residues enter the two-loop action RG combinations
with a factor four. For the quartic residue the factor is $2(3-1)=4$,
because its coupling degree is three and $c_B$ carries $\mu^{2\eps}$.

A constant-field potential fixes action combinations, not the wavefunction
counterterm separately. The additional local calculation is the massless
momentum-dependent sunset. Successive one-loop convolutions give
\begin{equation}
 J(p)=\frac{(p^2)^{1-2\eps}}{(4\pi)^{4-2\eps}}
 \frac{\Gamma(1-\eps)^3\Gamma(-1+2\eps)}{\Gamma(3-3\eps)}
 =-\frac{p^2}{4K^2\eps}+O(1).
\end{equation}
For the vertex tensor
$2c(\delta_{ab}\delta_{cd}+\delta_{ac}\delta_{bd}+\delta_{ad}\delta_{bc})$,
the sunset contraction is $2(N+2)c^2$, including the symmetry factor $1/6$.
Its 1PI sign is negative, giving the divergent kinetic correction
$+(N+2)c^2p^2/(2K^2\eps)$. Consequently,
\begin{equation}
 Z_\phi=1-\frac{(N+2)c^2}{2K^2\eps}+O(c^3),\qquad
 \gamma_\phi=\frac12\mu\frac{\dd\log Z_\phi}{\dd\mu}
 =\frac{(N+2)c^2}{K^2}+O(c^3).
\end{equation}
Reading the simple poles and including field normalisation yields
\begin{align}\label{pb:betas}
 \beta_c&=\frac{2(N+8)c^2}{K}-\frac{12(3N+14)c^3}{K^2}+O(c^4),\\
 \beta_{m^2}&=\frac{2(N+2)cm^2}{K}-\frac{10(N+2)c^2m^2}{K^2}+O(c^3m^2),\\
 \beta_\xi&=\frac{2(N+2)c}{K}(\xi-\tfrac16)
 +\frac{(N+2)c^2}{K^2}\left[-10(\xi-\tfrac16)+\frac23\right]+O(c^3),\\
 \gamma_\phi&=\frac{(N+2)c^2}{K^2}+O(c^3).
\end{align}
Thus $\xi=1/6$ is not an exact two-loop zero of $\beta_\xi$ in this
convention. For the mass rather than its square,
$\gamma_m=\beta_{m^2}/(2m^2)$. In the alternative quartic normalisation
$\lambda(\phi^2)^2/4!$, the coupling is $\lambda=6c$.

\section{Conclusions}\label{sec:conclusions}
We have constructed the renormalised two-loop effective potential of the
$O(N)$-symmetric scalar theory on four-dimensional de Sitter space, retaining
the full dependence on the field-dependent masses and curvature. The central
ingredient is the watermelon integral, obtained from the finite-difference
equations of Ref.~\cite{BCM}. We have proved the boundary condition selecting
the physical solution by comparison with conformally coupled spectator lines.
The resulting representation isolates the ultraviolet poles and gives the
finite part as an absolutely convergent series.

We have tracked the regulator dependence of the spectral labels when the
field-dependent masses and scalar curvature are held fixed. The associated
derivative terms cancel after the forest subtraction. Combining the
renormalised watermelon with the double-tadpole diagrams gives the two-loop
potential, and the local counterterms determine the matter-sector beta
functions once the wavefunction renormalisation is supplied by a separate
momentum-dependent calculation.

The result is perturbative: exact mass-to-curvature dependence does not
remove infrared limitations near propagator zero modes. The corresponding
anti-de Sitter calculation and extensions to higher-loop vacuum diagrams are
natural applications of the same finite-difference framework.

\appendix
\section{Lowering the weight: only linear digammas remain}\label{sec:abel}
The summand of Eq.~\eqref{eq:C0replacement-shortseries} is quadratic in
digammas. An exchange of summation removes these squares because
Eq.~\eqref{eq:tel} supplies both tails in closed form. In this appendix use
\begin{align}
 Y_i&=Y_{q_i},\quad a_i=a_{q_i},\quad b_i=b_{q_i},\quad
 v_n=q_3+1+2n,\quad g_n=g_{v_n},\\
 S&=Y_1+Y_2,\qquad D=Y_1-Y_2,\qquad P=Y_1Y_2.
\end{align}
The tails are
\begin{equation}\label{eq:tails}
 \mathsf T_n=\sum_{m\ge n}K_m
 =\frac{-4\Pi-W_n(Y_{v_n-1}-S)}{4P},\qquad
 \mathsf U_n=\sum_{m\ge n}K_mY_{v_m}=S\mathsf T_n-W_n.
\end{equation}
Set
\begin{equation}
 A=a_1+a_2,\quad B=b_1+b_2,\quad
 C_{\rm A}=-2(Y_1a_2+Y_2a_1),\quad
 E=S-D(a_1-a_2),\quad J_n=L_n+g_n.
\end{equation}
The appendix coefficient $C_{\rm A}$ is distinct from the finite-part
function $C(\bq)$. Define
\begin{align}\label{eq:U}
 N_n&=-\frac S2[J_n^2+{\Lambda_n}+2\Omega_{q_3}]+AJ_n+B,\\
 V_n&=\frac{D^2}{2}[J_n^2+{\Lambda_n}+2\Omega_{q_3}]+EJ_n+C_{\rm A}L_n
 +P(L_n^2+{\Lambda_n})\nonumber\\
 &\quad+2a_1a_2-A-S(1+B).\label{eq:V}
\end{align}
The bracket in Eq.~\eqref{eq:C0replacement-shortseries} is
$Y_{v_n}N_n+V_n$. All increments are rational:
\begin{align}\label{eq:incr}
 p_n&=\frac12\sum_{s,t}
 \left[\frac1{n+A_{st}^{(4)}}+\frac1{n+B_{st}^{(4)}+1}\right],\nonumber\\
 r_n&=-\frac2{v_n+1/2}-\frac2{v_n+3/2},\qquad t_n=p_n+r_n,\\
 {\delta\Lambda_n}&={\Lambda_{n+1}}-{\Lambda_n}
 =\frac14\sum_{s,t}\left[-\frac1{(n+A_{st}^{(4)})^2}
 +\frac1{(n+B_{st}^{(4)}+1)^2}\right].\nonumber
\end{align}
Hence
\begin{align}\label{eq:dU}
 \delta N_n&=(A-SJ_n)t_n-\frac S2(t_n^2+{\delta\Lambda_n}),\\
 \delta V_n&=(D^2J_n+E)t_n+(C_{\rm A}+2PL_n)p_n
 +\frac{D^2}{2}(t_n^2+{\delta\Lambda_n})+P(p_n^2+{\delta\Lambda_n}).\label{eq:dV}
\end{align}
\begin{theorem}[Lower-weight form]\label{thm:abel}
The sum in Eq.~\eqref{eq:C0replacement-shortseries} equals
\begin{equation}\label{eq:abel}
 \mathsf U_0N_0+\mathsf T_0V_0
 +\sum_{n=0}^\infty
 [\mathsf N_{n+1}\delta N_n+\mathsf T_{n+1}\delta V_n].
\end{equation}
Its summand is $O(n^{-3}\log n)$ and contains only linear $J_n,L_n$ and
rational functions of $n$: there is no moving trigamma or squared digamma.
\end{theorem}
\begin{proof}
Use $\mathsf T_n-\mathsf T_{n+1}=K_n$ and
$\mathsf U_n-\mathsf U_{n+1}=K_nY_{v_n}$, and apply Abel summation to $N$
and $V$ separately. The endpoint
$\mathsf U_NN_N+\mathsf T_NV_N$ vanishes because
$\mathsf U_N=O(N^{-2})$, $\mathsf T_N=O(N^{-4})$, and
$N_N,V_N=O(\log^2N)$.
\end{proof}


\begin{thebibliography}{99}
\bibitem{ColemanWeinberg1973}
S.~Coleman and E.~Weinberg,
``Radiative corrections as the origin of spontaneous symmetry breaking,''
Phys.\ Rev.\ D \textbf{7} (1973) 1888.
\href{https://doi.org/10.1103/PhysRevD.7.1888}{doi:10.1103/PhysRevD.7.1888}.

\bibitem{Jackiw1974}
R.~Jackiw, ``Functional evaluation of the effective potential,''
Phys.\ Rev.\ D \textbf{9} (1974) 1686.
\href{https://doi.org/10.1103/PhysRevD.9.1686}{doi:10.1103/PhysRevD.9.1686}.

\bibitem{Shore}
G.~M.~Shore,
``Radiatively Induced Spontaneous Symmetry Breaking and Phase Transitions in Curved Space-Time,''
Annals Phys.\ \textbf{128} (1980) 376.
\href{https://doi.org/10.1016/0003-4916(80)90326-7}{doi:10.1016/0003-4916(80)90326-7}.

\bibitem{Allen1983}
B.~Allen, ``Phase Transitions in de Sitter Space,''
Nucl.\ Phys.\ B \textbf{226} (1983) 228--252.
\href{https://doi.org/10.1016/0550-3213(83)90470-4}{doi:10.1016/0550-3213(83)90470-4}.

\bibitem{Starobinsky1986}
A.~A.~Starobinsky, ``Stochastic de Sitter (Inflationary) Stage in the Early Universe,''
in \emph{Field Theory, Quantum Gravity and Strings},
Lect.\ Notes Phys.\ \textbf{246} (1986) 107.
\href{https://link.springer.com/chapter/10.1007/3-540-16452-9_6}{doi:10.1007/3-540-16452-9\_6}.

\bibitem{StarobinskyYokoyama1994}
A.~A.~Starobinsky and J.~Yokoyama,
``Equilibrium State of a Self-Interacting Scalar Field in the de Sitter Background,''
Phys.\ Rev.\ D \textbf{50} (1994) 6357.
\href{https://doi.org/10.1103/PhysRevD.50.6357}{doi:10.1103/PhysRevD.50.6357}.

\bibitem{MarkkanenRajantieStopyra2018}
T.~Markkanen, A.~Rajantie and S.~Stopyra,
``Cosmological Aspects of Higgs Vacuum Metastability,''
Front.\ Astron.\ Space Sci.\ \textbf{5} (2018) 40.
\href{https://doi.org/10.3389/fspas.2018.00040}{doi:10.3389/fspas.2018.00040}.

\bibitem{Serreau2011}
J.~Serreau, ``Effective Potential for Quantum Scalar Fields in a de Sitter Geometry,''
Phys.\ Rev.\ Lett.\ \textbf{107} (2011) 191103.
\href{https://doi.org/10.1103/PhysRevLett.107.191103}{doi:10.1103/PhysRevLett.107.191103}.

\bibitem{LopezNacirMazzitelliTrombetta2016}
D.~L\'opez Nacir, F.~D.~Mazzitelli and L.~G.~Trombetta,
``$O(N)$ model in Euclidean de Sitter space: beyond the leading infrared approximation,''
JHEP \textbf{09} (2016) 117.
\href{https://doi.org/10.1007/JHEP09(2016)117}{doi:10.1007/JHEP09(2016)117}.

\bibitem{ACM}
A.~Bonanno, S.~L.~Cacciatori and U.~Moschella,
``Effective Potentials for de Sitter and Anti-de Sitter Quantum Fields,''
Symmetry \textbf{18} (2026) no.~5, 801.
\href{https://doi.org/10.3390/sym18050801}{doi:10.3390/sym18050801}.

\bibitem{FordJackJones1992}
C.~Ford, I.~Jack and D.~R.~T.~Jones,
``The Standard Model Effective Potential at Two Loops,''
Nucl.\ Phys.\ B \textbf{387} (1992) 373;
Erratum: Nucl.\ Phys.\ B \textbf{504} (1997) 551.
\href{https://doi.org/10.1016/0550-3213(92)90165-8}{doi:10.1016/0550-3213(92)90165-8}.

\bibitem{Chetyrkin:1981qh}
K.~G.~Chetyrkin and F.~V.~Tkachov,
``Integration by parts: The algorithm to calculate $\beta$-functions in 4 loops,''
Nucl.\ Phys.\ B \textbf{192} (1981) 159.
\href{https://doi.org/10.1016/0550-3213(81)90199-1}{doi:10.1016/0550-3213(81)90199-1}.

\bibitem{Laporta:2000dsw}
S.~Laporta,
``High-precision calculation of multiloop Feynman integrals by difference equations,''
Int.\ J.\ Mod.\ Phys.\ A \textbf{15} (2000) 5087.
\href{https://doi.org/10.1142/S0217751X00002159}{doi:10.1142/S0217751X00002159}.

\bibitem{CEM2024}
S.~L.~Cacciatori, H.~Epstein and U.~Moschella,
``Loops in de Sitter space,'' JHEP \textbf{07} (2024) 182.
\href{https://doi.org/10.1007/JHEP07(2024)182}{doi:10.1007/JHEP07(2024)182}.

\bibitem{bros1}
J.~Bros, H.~Epstein, M.~Gaudin, U.~Moschella and V.~Pasquier,
``Triangular invariants, three-point functions and particle stability on the de Sitter universe,''
Commun.\ Math.\ Phys.\ \textbf{295} (2010) 261.
\href{https://doi.org/10.1007/s00220-009-0875-4}{doi:10.1007/s00220-009-0875-4}.

\bibitem{bros2}
J.~Bros, H.~Epstein, M.~Gaudin, U.~Moschella and V.~Pasquier,
``Anti de Sitter Quantum Field Theory and a New Class of Hypergeometric Identities,''
Commun.\ Math.\ Phys.\ \textbf{309} (2012) 255.
\href{https://doi.org/10.1007/s00220-011-1372-0}{doi:10.1007/s00220-011-1372-0}.

\bibitem{Arkani}
N.~Arkani-Hamed, D.~Baumann, A.~Hillman, A.~Joyce, H.~Lee and G.~L.~Pimentel,
``Differential Equations for Cosmological Correlators,''
JHEP \textbf{09} (2025) 009.
\href{https://doi.org/10.1007/JHEP09(2025)009}{doi:10.1007/JHEP09(2025)009}.

\bibitem{Chen}
J.~Chen and B.~Feng, 
``Towards systematic evaluation of de Sitter correlators via Generalized Integration-By-Parts relations,''
JHEP \textbf{06} (2024) 199.
\href{https://doi.org/10.1007/JHEP06(2024)199}{doi:10.1007/JHEP06(2024)199}.

\bibitem{ChenFengQinTao}
J.~Chen, B.~Feng, Z.~Qin and Y.~X.~Tao,
``Loop integrals in de Sitter spacetime: The parity-split IBP system and d log-form
differential equations,''
\href{https://doi.org/10.48550/arXiv.2604.14549}{arXiv:2604.14549 (2026)}

\bibitem{DiakonovMorozov}
D.~Diakonov and A.~Morozov,
Phys.\ Lett.\ B \textbf{858} (2024) 139079, arXiv:2408.15724.
\href{https://doi.org/10.1016/j.physletb.2024.139079}{doi:10.1016/j.physletb.2024.139079}.

\bibitem{CostaHansen2018}
M.~S.~Costa and T.~Hansen, ``AdS Weight Shifting Operators,''
JHEP \textbf{09} (2018) 040.
\href{https://doi.org/10.1007/JHEP09(2018)040}{doi:10.1007/JHEP09(2018)040}.

\bibitem{Alaverdian}
M.~Alaverdian, A.~Herderschee, R.~Roiban and F.~Teng,
JHEP \textbf{12} (2024) 070, arXiv:2406.04186.

\bibitem{Carmi2020}
D.~Carmi, ``Loops in AdS: From the Spectral Representation to Position Space,''
JHEP \textbf{06} (2020) 049.
\href{https://doi.org/10.1007/JHEP06(2020)049}{doi:10.1007/JHEP06(2020)049}.

\bibitem{Carmi2021}
D.~Carmi, ``Loops in AdS: from the spectral representation to position space. Part II,''
JHEP \textbf{07} (2021) 186.
\href{https://doi.org/10.1007/JHEP07(2021)186}{doi:10.1007/JHEP07(2021)186}.

\bibitem{Carmi2024}
D.~Carmi, ``Loops in AdS: from the spectral representation to position space. Part III,''
JHEP \textbf{08} (2024) 193.
\href{https://doi.org/10.1007/JHEP08(2024)193}{doi:10.1007/JHEP08(2024)193}.

\bibitem{BCM}
B.~A.~Bassett, {S.~L.~Cacciatori} and U.~Moschella,
``Finite Spectral Recurrences for Loop Integrals in de Sitter and Anti-de Sitter Space,''
\href{https://arxiv.org/abs/2610.28183}{arXiv:2610.28183} [hep-th].

\bibitem{bgm}
J.~Bros, U.~Moschella and J.~P.~Gazeau,
``Quantum field theory in the de Sitter universe,''
Phys.\ Rev.\ Lett.\ \textbf{73} (1994) 1746--1749.
\href{https://doi.org/10.1103/PhysRevLett.73.1746}{doi:10.1103/PhysRevLett.73.1746}.

\bibitem{bm}
J.~Bros and U.~Moschella,
``Two point functions and quantum fields in de Sitter universe,''
Rev.\ Math.\ Phys.\ \textbf{8} (1996) 327--392.
\href{https://doi.org/10.1142/S0129055X96000123}{doi:10.1142/S0129055X96000123}.

\bibitem{Tachyons1}
J.~Bros, H.~Epstein and U.~Moschella,
``Scalar tachyons in the de Sitter universe,''
Lett.\ Math.\ Phys.\ \textbf{93} (2010) 203--211.
\href{https://doi.org/10.1007/s11005-010-0406-4}{doi:10.1007/s11005-010-0406-4}.

\bibitem{Tachyons2}
H.~Epstein and U.~Moschella, ``de Sitter tachyons and related topics,''
Commun.\ Math.\ Phys.\ \textbf{336} (2015) no.~1, 381--430.
\href{https://doi.org/10.1007/s00220-015-2308-x}{doi:10.1007/s00220-015-2308-x}.

\bibitem{bateman}
A.~Erd\'elyi (ed.), \emph{Higher Transcendental Functions},
Bateman Manuscript Project, Vol.~I, McGraw-Hill, New York (1953).

\bibitem{VicenteGarcia-Consuegra:2025two}
L.~Vicente Garc{\'\i}a-Consuegra and A.~Rajantie,
``Scalar field effective potentials in de Sitter spacetime,''
Phys.\ Rev.\ D \textbf{113} (2026) 125016.
\href{https://doi.org/10.1103/tkw7-dkq6}{doi:10.1103/tkw7-dkq6}.

\bibitem{Szego1975}
G.~Szeg\H{o}, \emph{Orthogonal Polynomials}, 4th ed.,
American Mathematical Society Colloquium Publications, Vol.~23,
American Mathematical Society, Providence, RI (1975).

\bibitem{DaiXu2013}
F.~Dai and Y.~Xu, \emph{Approximation Theory and Harmonic Analysis on Spheres and Balls},
Springer Monographs in Mathematics, Springer, New York (2013).

\bibitem{DLMF}
NIST Digital Library of Mathematical Functions,
Secs.~\href{https://dlmf.nist.gov/15.5}{15.5}, and
\href{https://dlmf.nist.gov/18.9}{18.9}, contiguous relations and
orthogonal-polynomial recurrences.

\bibitem{DLMF1}
NIST Digital Library of Mathematical Functions,
Secs.~\href{https://dlmf.nist.gov/18.10}{18.10}, integral representations.

\end{thebibliography}
\end{document}